\pdfoutput=1
\documentclass[11pt]{article} 
\usepackage{aaai2027_copy} 
\nocopyright
\usepackage{amsthm}
\usepackage{amsmath}

\usepackage[disable]{todonotes}
\usepackage[hyphens]{url}  
\usepackage{graphicx} 
\usepackage{natbib}  
\usepackage{caption} 
\usepackage[ruled,vlined]{algorithm2e}
\usepackage{amsfonts}
\usepackage{algorithmic}
\usepackage{url}
\newtheorem{theorem}{Theorem}

\newtheorem{corollary}{Corollary}
\newtheorem{definition}{Definition}

\newtheorem{lemma}{Lemma}

\newif\ifarxiv
\arxivtrue

\newcommand{\OPT}{\mathrm{OPT}}

\newcommand{\ex}{\mathbb{E}}

\newcommand{\opt}{\ensuremath{\mathrm{OPT}}\xspace}

\newcommand{\jens}[1]{\todo[backgroundcolor=blue!25,bordercolor=blue]{J:~#1}}

\newcommand{\raj}[1]{\todo[backgroundcolor=green!25,bordercolor=green]{R:~#1}}

\usepackage{float}
\usepackage{newfloat}
\usepackage{listings}
\DeclareCaptionStyle{ruled}{labelfont=normalfont,labelsep=colon,strut=off} 
\floatstyle{ruled}
\newfloat{listing}{tb}{lst}{}
\floatname{listing}{Listing}
\usepackage{amsthm}

\usepackage{cleveref}
\Crefname{theorem}{Theorem}{Theorems}
\Crefname{lemma}{Lemma}{Lemmas}
\Crefname{corollary}{Corollary}{Corollaries}
\Crefname{proposition}{Proposition}{Propositions}
\Crefname{observation}{Observation}{Observations}
\Crefname{claim}{Claim}{Claims}
\Crefname{fact}{Fact}{Facts}
\Crefname{assumption}{Assumption}{Assumptions}
\crefname{thm}{Theorem}{Theorems}

\crefname{thm}{Theorem}{Theorems}
\crefname{lem}{Lemma}{Lemmas}
\usepackage{booktabs}
\usepackage{nicefrac}
\usepackage{thm-restate}
\title{Correlation Clustering with Random Partial Information}
\author {
    Rajath Rao K.N \textsuperscript{\rm 1},
    Jens  Schl\"oter\textsuperscript{\rm 2},
    Sami Davies \textsuperscript{\rm 3},
    Amira Ouchene \textsuperscript{\rm 1} ,
    Yasamin Nazari \textsuperscript{\rm 1,\rm 2}
}
\affiliations {
    \textsuperscript{\rm 1}Vrije Universiteit Amsterdam, The Netherlands\\
    \textsuperscript{\rm 2}CWI Amsterdam, The Netherlands \\
    \textsuperscript{\rm 3}Department of EECS, UC Berkeley
}

\begin{document}

\maketitle
\begin{abstract}
 Correlation clustering is a fundamental unsupervised learning problem. On complete graphs, both the min-disagreement and min-max objectives admit constant-factor
approximations, yet on general (non-complete) graphs, the best guarantees blow up to
$O(\log n)$ and $O(\sqrt{n})$.
This gap between the two regimes motivates the following question: are there classes of incomplete graphs that circumvent
the lower bounds on general graphs and admit approximation guarantees
approaching those attainable on complete graphs?
We study a natural class of graphs obtained by randomly subsampling a complete signed
graph $G$, where each edge is independently deleted with probability $q$.
For such graph instances both for the min-max and the min-disagreement objectives, we prove approximation guarantees (depending on $q$) that are substantially better than the bounds achievable for general graphs.
We supplement our theoretical results with experiments that also suggest that the approximation ratios of our algorithm are
close to those of the complete graph and better than the worst-case bounds
for general (non-complete) graphs. 
\end{abstract}

\section{Introduction}

Correlation clustering is a fundamental problem in unsupervised learning that has received much attention over the past two decades. The input is an undirected graph $G=(V,E)$, with edges labeled either positive $(+)$ or negative $(-)$. A positive edge indicates its endpoints are similar, while a negative edge indicates they are dissimilar. 
A valid clustering $\mathcal{C}$ is a partition of $V$ that can have any number of clusters. An edge is in \emph{disagreement} with respect to $\mathcal{C}$ if it is a positive edge whose endpoints are in different clusters of $\mathcal{C}$, or if it is a negative edge whose endpoints are in the same cluster of $\mathcal{C}$. Our goal is to find the optimal clustering (using any number of clusters) with respect to an objective function defined on such disagreements.
The problem is most well-studied in the setting when $G$ is a complete graph. For the objective of finding a clustering that minimizes the total number of edges in disagreement, called the \emph{min-disagreement} objective, the problem is NP-hard to approximate within a factor of $24/23$ \cite{cao2024}. There are many constant factor approximation algorithms for this setting \cite{bansal2004correlation, charikar2005clustering, ACN-pivot, chawla2015near, cohen2022correlation, Cohen-AddadLPTY24}, with the current best factor being 1.485 \cite{cao2024}. 

The state of affairs is similar for the objective of finding a clustering that minimizes the \textit{maximum} number of disagreements incident to any vertex, called the \emph{min-max} objective.
The min-max objective is NP-hard to approximate within a factor better than $4/3$, but there are several constant factor approximation algorithms \cite{puleo2015correlation, CGS17, KMZ19, DMN23, heidrich2024}, and the best factor is 3 \cite{cao2025min}. 

For both objectives, the problem is much less understood when $G$ is not a complete graph, a setting also known as the \emph{general graph} or partial information setting. For min-disagreement, there is an $O(\log n)$ approximation algorithm (for $n = |V|$), and improving this to $o(\log n)$ would necessitate an improved approximation for minimum multicut \cite{demaine2006correlation}. For the min-max objective, there are $\widetilde{O}(\sqrt{n})$ approximation\footnote{We note there is no logarithmic term in the work of \citet{CGS17}; the approximation factor is $O(\sqrt{n})$.} algorithms \cite{CGS17, KMZ19}, and natural LPs for the problem have large integrality gaps, making further algorithmic improvement difficult. On the other hand, the min-max objective is only known to be NP-hard to approximate within a factor better than $2-\varepsilon$ \cite{KMZ19}.

Therefore for both objectives, we have a large disparity between the best approximation factors in the complete graph versus general graph settings. Our motivation is to find a natural graph class that is not complete, but can attain better approximation factor. The simplest such class that we consider is one where a random set of edges are missing, either due to lack of information or the connection structure itself.

\subsection{Our contribution}
In the following results, we assume that $G= (V,E)$ is a complete signed graph on $n$ nodes. 
The graph $G'$ is constructed from $G$ by independently deleting each edge in $G$ with probability $q$. Formally, we write $\mathcal{S}_{q}(G)$ as the distribution over edge-induced subgraphs of $G$ obtained by independently deleting each edge with probability $q$, then sample $G' \sim \mathcal{S}_{q}(G)$.



 Our results make different assumptions on what information the algorithm has access to. We may assume access only to the original complete graph $G$. This setting can be motivated by a notion of robustness, where the guarantee holds even after some edges are removed. The second and most natural assumption for our main motivation is when we have access only to $G' \sim \mathcal{S}_{q}(G)$. Finally, we show more general guarantees when we have access both to $G$ and $G'$, which is more for getting theoretical insights that could be useful for possible future extensions of our result.

\paragraph{Min-Max Correlation Clustering.} Our first result is for min-max correlation clustering under random deletions.  Let $\opt_{\infty}(G')$ denote the cost of an optimal solution for the min-max objective on graph $G'$.

\begin{theorem}
\label{lemma:min_max_correlation_1}
    Given $G' \sim \mathcal{S}_{q}(G)$, there exists an algorithm that, with high probability, returns a clustering for $G'$ that is $O\!\left(\nicefrac{\log n}{1-q}\right)$-approximate. 
\end{theorem}

\raj{Is the approximation ratio correct? is it $\frac{\log n}{(1-q)}$}\jens{Can you elaborate why you think it is incorrect? It seems to match the proof...}
Throughout, we use \emph{with high probability} to  mean with probability at least $1-\nicefrac{1}{n}$. The randomness in this statement is only over the deletions, and the algorithm is deterministic. 
This approximation guarantee is substantially better than the best known of $O(\sqrt{n})$ for general graphs by \citet{CGS17}, and while their result is LP-based, our algorithm is combinatorial.
Our inspiration for the algorithm is a result of \citet{heidrich2024}: it shows that on complete graphs, vertices with large overlapping positive neighborhoods must be in the same cluster of any optimal min-max solution, and vertices with very different positive neighborhoods must be in different clusters of any optimal min-max solution.
Also, using randomized tools such as Johnson-Lindenstrauss (as used by \citet{cao2025min}) the running time can be improved to $\widetilde{O}(n^2)$.

\paragraph{Min-disagreement.}

We first show that if we have access to \emph{only} $G$, we can use \emph{any} $\alpha$-approximation algorithm for the complete graph $G$ to obtain an $O(\nicefrac{\alpha}{(1-q)^2})$-approximation algorithm for the graph $G' \sim \mathcal{S}_q(G)$ \textit{in expectation} over the randomness of the deletions. This result is based on the observation that the expected cost of an optimal solution for $G'$ does not decrease too fast with the deletion probability. More precisely, we prove $\ex[\opt(G')] \ge \Omega((1-q)^3 \cdot \opt(G))$, where  $\opt(G')$ and $\opt(G)$  denote the cost of optimal solutions for the min-disagreement objective on $G'$ and $G$.

We show that a similar property even holds with  high probability when $\opt(G) \in \Omega(\log n)$. Using this stronger statement, we show that, given access to $G$ and assuming $\opt(G) \in \Omega(\log n)$, each $\alpha$-approximation algorithm for the complete graph yields an $O(\nicefrac{\alpha}{(1-q)^2})$-approximation for $G'$ with high probability.

Next, we move on to the technically more challenging setting where the algorithm \emph{only} has access to $G' \sim \mathcal{S}_q(G)$. This turns out to be difficult to show in general, but we get better guarantees in certain regimes of $\opt(G)$ and depending on what information we have access to.
If $\opt(G)$ is small, we can show the following:

\begin{theorem}
\label{our_result_low-cost-direct_}
    If $\OPT(G) = O(\log^\gamma n)$ for any constant $\gamma$, then there is a polynomial time algorithm that, given only $G' \sim \mathcal{S}_{q}(G)$, computes an $O(\log (\frac{\log n}{(1-q)^2}))$-approximation for $G'$ with high probability. 
\end{theorem}

This is our main technical contribution. For constant $q$, the approximation ratio is significantly better than the guarantee of $O(\log n)$ for general graphs.
We prove~\Cref{our_result_low-cost-direct_} by observing that in this setting, every optimal solution for $G$ or $G'$ consists of only a few clusters with disagreements. 
Additionally, we show that we can accurately identify large optimal clusters (size $polylog(n)$) of $G'$ with high probability. In fact, we show that these large clusters are, with high probability, part of the optimal solution for both $G$ and $G'$. This implies that sufficiently large optimal clusters of $G$ can be recovered from $G'$ with high probability.
After finding and removing disagreement-free clusters and large clusters, we will be left with a smaller instance with only $n_s= polylog(n)$ vertices on which we can run the algorithm of \citet{demaine2006correlation} for general graphs to get an $O(\log (n_s)) = O(\log \log n)$ approximation for $G'$.

We believe that the assumption in~\Cref{our_result_low-cost-direct_} that OPT$(G)$ is small is a natural setting. As will be outlined in~\Cref{sec: relatedwork}, the algorithmic task of recovering a perfect ground truth clustering, that is, an instance with cost zero, after noisy perturbations has received much attention in the literature. The setting of~\Cref{our_result_low-cost-direct_} can be considered as a variant of this task, where the goal is to (approximately) recover a slightly imperfect ground truth clustering, that is, an instance with small cost $O(\log n)$, after noisy perturbations in the form of edge deletions.






We can also obtain the following result for the full range of $\opt(G)$ if we are given $G$ and the realized $G'$.

\begin{restatable}{theorem}{thmBothGraphs}
    \label{thm:l1:both:graphs}
    Given $G$ and $G' \sim \mathcal{S}_{q}(G)$, there exists a polynomial time algorithm that computes an  $O(\max\{\nicefrac{1}{(1-q)^2}, \log\log n + \log(\nicefrac{1}{1-q})\})$-approximation for $G'$ with high probability.
\end{restatable}

\paragraph{Experimental results.} We complement our theoretical findings with
empirical results, on random deletion instances obtained from the
Facebook ego-network dataset and on  planted-clique instances based on the stochastic block model. 

Our experiments for min-max correlation clustering, suggest that~\Cref{lemma:min_max_correlation_1} is predictive of practice, as our empirical approximations are at least as good as our theoretical guarantee of $O\!\left(\nicefrac{\log n}{1-q}\right)$. 

For the min-disagreement objective, we study the empirical approximation ratio of \textsc{Pivot}~\cite{ACN-pivot} when run on graph $G'$. Our theoretical results show that the clustering computed by \textsc{Pivot} for the original complete graph $G$ yields an expected $O(\nicefrac{\alpha}{(1-q)^2})$-approximation for $G'$. The empirical results suggest that this guarantee may also hold for the clustering computed by \textsc{Pivot} on $G'$, which we do not show theoretically.

\ifarxiv
Proofs are given alongside their statements below; further details on our experiments (including source code) are provided in the supplementary materials.
\else
Omitted proofs and details on our experiments (including source code) are provided in the supplementary materials.
\fi

\subsection{Related work}\label{sec: relatedwork}
To the best of our knowledge, for the min-max objective, there are no known special classes of non-complete graphs that admit approximation algorithms better than $O(\sqrt{n})$. 
We note the min-max objective is a measure of fairness, and other fairness-motivated and local objectives have been considered \cite{friggstad2021fair, davies2023one}, as well as fairness constraints \cite{ahmadi2020fair, ahmadian2020fair}.
There is also  work on general graphs for the min-disagreement objective with fairness constraints \cite{schwartz2022fair}.

For the min-disagreement objective, \citet{demaine2006correlation} showed an $O(r^3)$ approximation for $K_{r,r}$-minor-free graphs, and \citet{ParsaeiMajd26} identified a class of graphs for which there is a 2-approximation by defining a small set of forbidden signed subgraphs.

Other work on the min-disagreement objective on subclasses of general graphs holds only for classes of graphs arising from a noisy and partially observed planted partition model---these settings model the problem of finding an underlying ground truth clustering in the presence of random noise. 
\citet{ChenJSX14} study under what conditions one can find an optimal clustering on an Erd{\"o}s-Renyi random graph of observed edges, where the signs of the observed edges are according to a planted partition model, except the signs are also corrupted i.i.d. 
An even more general model is studied by \citet{makarychev2015correlation}; here, an adversary chooses an arbitrary (unsigned) graph and a ground truth clustering on that graph, then each edge in the graph is sampled i.i.d., and the adversary can flip the signs of any of these sampled edges.
These last two models setting were inspired by work on complete graphs arising from a noisy planted partition model \cite{Ben-DorSY99, bansal2004correlation,mathieu2010correlation}
reconstruct the ground truth clustering as well as approximation algorithms.
Similar models (though less relevant to our study) have remained popular over the last decade \cite{ChierichettiPRT22, WangWS22}.



\subsection{Preliminaries}
Recall $G = (V, E)$ denotes a
signed graph on $n = |V|$ vertices. 
We use $G' = (V, E' )$ to refer to a subgraph of $G$ that is generated from $G$ by independently deleting each edge $e \in E $ with probability $q$, where we formally write $G' \sim \mathcal{S}_q(G)$. Note when defining such $G'$, the graph $G$ is always complete. We let $E^+$ denote the positive edges of $G$ and $E^-$ the negative, so $E = E^+ \cup E^-$, and similarly $E' = E'^+ \cup E'^-$.

For $i \in \{+,-\}$ and $v \in V$, define the \emph{positive and negative neighborhood} of $v$ in $G$ as $N^i_G(v) := \{ w \in V \mid \{w,v\} \in E^i\}$. Analogously, we use $N^+_{G'}(v)$ and $N^-_{G'}(v)$ to refer to the neighborhoods in $G'$.
We use $\mathcal{C}$ to denote a clustering, which is a partition of the vertex set $V$. For a clustering
$\mathcal{C}$ and a vertex $v \in V$, 
let $[v]_{\mathcal{C}}$ be the cluster of $v$ in $\mathcal{C}$.
Then we define
\[
    \text{cost}_G(\mathcal{C},v) = |N^+_G(v) \setminus [v]_{\mathcal{C}}| + |N^-_G(v) \cap [v]_{\mathcal{C}}|,
\]
as the number of disagreements incident to $v$ in $G$ with respect to $\mathcal{C}$. We use all of the same notation for the analogous terms in $G'$. If the graph is clear from context, we omit the subscript in the neighborhood notation or in the cost notation.

If an edge $e \in E$ satisfies $e \in E' := E'^+ \cup E'^-$, then we say that $e$ \emph{survives}. We write $X_e \sim \mathrm{Ber}(1-q)$ for the indicator that edge $e$ survives. 


\section{Min-disagreement Objective} 
\label{sec:min:disagreement}

In this section, we prove our results for the objective of minimizing the total number of disagreements, Theorems \ref{thm:l1:both:graphs} and \ref{low-cost-direct}. We let $\text{cost}_{G'}(\mathcal{C}) = \nicefrac12 \cdot \sum_{v \in V} \text{cost}_{G'}(\mathcal{C},v),$ so the cost of an optimal clustering on $G'$ is denoted by
$
    \text{OPT}(G') = \min_{\mathcal{C}} \text{cost}_{G'}(\mathcal{C}),$ and analogously for $G$.

We start by proving properties of the (expected) optimum and providing approximation results that hold in expectation. Then, we prove upper bounds holding with high probability.

\subsection{Expected upper bounds}
\label{sec:exp:upper:bounds}

For $G' \sim \mathcal{S}_q(G)$, we derive some useful properties of $\ex[\opt(G')]$ as a warm-up. In particular, we compare $\ex[\opt(G')]$  to $\opt(G)$. 
This immediately implies an expected approximation factor for the case where $G$ is known to the algorithm, and will later be useful for proving upper bounds with high probability.
To this end, we rely on well-established connections between  $\opt(G)$ and the number of \emph{bad triangles} in $G$. 

\begin{definition}[Bad triangles]
    A triangle $T$ in $G=(V,E^+ \cup E^-)$ is \emph{bad} if $|T \cap E^+| = 2$ and $|T \cap E^-| = 1$. We use $T(G)$ to refer to the maximum set of edge disjoint bad triangles in $G$, and define $\tau(G) := |T(G)|$.
\end{definition}

It is well-known that $\tau(G)$ lower bounds $\opt(G)$, even if $G$ is not complete~\cite{bansal2004correlation}. If $G$ is complete, then we even get $\opt(G) \in \Theta(\tau(G))$ by exploiting existing constant factor approximations that bound their costs against a triangle-based LP.

\begin{restatable}{lemma}{lemOptTriangleBound}
    \label{lem:opt:triangle:bound}
    For every signed (not necessarily complete) graph $G$, we have $\tau(G) \le \opt(G)$.
    If $G$ is complete, then $\opt(G) \le 9 \cdot \tau(G)$.
\end{restatable}

\ifarxiv
\begin{proof}
    Inequality $\tau(G) \le \opt(G)$ has, for example, been shown by~\citet{bansal2004correlation}. Hence, it remains to show that for complete $G$, we have $\opt(G) \le 9 \cdot \tau(G)$.

    We argue using the following LP-relaxation~\cite{ACN-pivot}, which uses $H$ to refer to the set of all (not only edge-disjoint) bad triangles.
    \begin{equation}
    \label{lower-bound-lp}
    \begin{aligned}
    \min \quad & \sum_{ab \in E} x_{ab} \\
    \text{s.t.} \quad
    & x_{ab} + x_{bc} + x_{ca} \ge 1, && \forall \{a,b,c\} \in H \quad , \\
    & x_{ab} \ge 0, && \forall ab \in E .
    \end{aligned}
    \end{equation}

    Let $\text{Cost(LP)}$ denote the cost of an optimal solution to LP~\eqref{lower-bound-lp}. We prove that (i) $\text{Cost(LP)} \leq 3 \cdot \tau(G)$ and (ii) $\opt(G) \le 3 \cdot \text{Cost}(LP)$. Combined, (i) and (ii) imply $\opt(G) \le 9 \cdot \tau(G)$.
    \begin{itemize}
        \item \textbf{Proof of (i)}: We prove $\text{Cost(LP)} \leq 3 \cdot \tau(G)$ by constructing a feasible LP solution $x'$ with objective value $3 \cdot \tau(G)$. Initialize $x'_{uv} = 0$ for all $uv \in E$.
        While there exists a violated bad-triangle constraint, pick such a triangle $\{a,b,c\}$
        and set $x'_{ab} = x'_{bc} = x'_{ca} = 1$. Let $T'$ be the triangles selected. Note that $T'$ is an edge-disjoint set of bad triangles and, therefore, $|T'| \le \tau(G)$. The solution $x'$ is feasible by construction. Hence, $ \text{Cost(LP)} \leq \text{Cost($x'$)} =3\cdot |T'| \leq 3 \cdot \tau(G)$.
        \item \textbf{Proof of (ii)}: The well-known analysis of the Pivot algorithm proves an approximation factor of $3$ by comparing against $\text{Cost}(LP)$. Hence, $\opt(G) \le 3 \cdot \text{Cost}(LP)$. \qedhere
    \end{itemize}
\end{proof}
\fi

The first part of~\Cref{lem:opt:triangle:bound} implies $\ex[\opt(G')] \ge \ex[\tau(G')]$. Furthermore, we can observe that $\ex[\tau(G')] \ge (1-q)^3 \tau(G)$ as each bad triangle in $T(G)$ is also present in $G'$ with an independent probability of $(1-q)^3$. Thus, using the second part of~\Cref{lem:opt:triangle:bound}, we get $\ex[\opt(G')] \ge \nicefrac{(1-q)^3}{9} \cdot \opt(G)$. 
An immediate corollary of this inequality is that every $\alpha$-approximate clustering for the complete graph $G$ is, in expectation, a $\nicefrac{9\cdot \alpha}{(1-q)^3}$-approximation for $G'$. With the following lemma, we prove a stronger bound.

 \begin{restatable}{lemma}{lemExpectedApprox} \label{lem:sparsification_G}
     Given access only to $G$, and a multiplicative $\alpha$-approximation clustering $\mathcal{C}$ for $G$, we have $\ex[\text{cost}_{G'}(\mathcal{C})] \le \frac{9 \cdot \alpha}{(1-q)^2} \cdot \ex[\opt(G')]$.
 \end{restatable}

\ifarxiv
\begin{proof}

We establish three inequalities.

\medskip\noindent\textbf{Inequality 1: $\ex[\OPT(G')] \geq \nicefrac{(1-q)^3}{9} \cdot \opt(G).$}
By~\Cref{lem:opt:triangle:bound}, $G$ has a set $T(G)$ of edge-disjoint bad triangles with $|T(G)| = \tau(G) \geq \nicefrac{\opt(G)}{9}$.
For each $t \in T(G)$, let $Y_t = \mathbf{1}[t \subseteq E']$, where $E'$ is the set of edges in $G'$.
The three edges of $t$ are distinct, so $\ex[Y_t] = (1-q)^3$.
Edge-disjointness of $T(G)$ makes the $Y_t$'s mutually independent.
Define $Y = \sum_{t \in T} Y_t$. Then,
\[
  \ex[Y] = (1-q)^3 \tau(G) \geq \nicefrac{(1-q)^3}{9} \cdot \opt(G).
\]
Surviving triangles, i.e., triangles $t \in T(G)$ with $t \subseteq E'$, are still edge-disjoint in $G'$.
Hence,~\Cref{lem:opt:triangle:bound} implies $\opt(G') \ge Y$. Taking expectations:
\begin{equation}
  \ex[\OPT(G')] \geq E[Y] \geq \nicefrac{(1-q)^3}{9} \cdot \opt(G). \label{eq:copt-upper}
\end{equation}

\medskip\noindent\textbf{Inequality 2:  $\mathrm{cost}_G(\mathcal{C}) \;\le\; \alpha\cdot\opt(G).$}
Since $\mathcal{C}$ is an $\alpha$-approximation on $G$,
\begin{equation}\label{eq:alpha-approx}
    \mathrm{cost}_G(\mathcal{C}) \;\le\; \alpha\cdot \opt(G).
\end{equation}

\medskip\noindent\textbf{Inequality 3: $\ex[\mathrm{cost}_{G'}(\mathcal{C})] = (1-q)\cdot \mathrm{cost}_G(\mathcal{C})$.} Using the fact that for each edge $e$ of $G$ we have $e \in E'$ with probability $(1-q)$ yields
\begin{equation} \label{eq:cost-C-G'}
    \ex[\text{cost}(\mathcal{C},G')] = (1-q) \cdot \text{cost}_G(\mathcal{C})
\end{equation}

We conclude the proof by substituting \eqref{eq:alpha-approx} into \eqref{eq:cost-C-G'} and then
applying \eqref{eq:copt-upper}:
\begin{align*}
    \ex[\mathrm{cost}_{G'}(\mathcal{C})]
        &\;\le\; (1-q)\,\alpha\cdot \opt(G) \\
        &\;\le\; (1-q)\,\alpha \cdot \nicefrac{9}{(1-q)^3}\cdot \ex[\opt(G')] \\
        &\;=\;   \nicefrac{9\alpha}{(1-q)^2}\cdot \ex[\opt(G')].
\end{align*}
\end{proof}
\fi

\Cref{lem:sparsification_G} establishes that if we have access to $G$, then we can compute an expected $\nicefrac{9\alpha}{(1-q)^2}$-approximation for $G'$ by just computing an $\alpha$-approximation for $G$. With the next lemma, we show a similar property that holds with high probability provided that $\opt(G)$ is sufficiently large.

\begin{restatable}{lemma}{lemOptHighProb}
    \label{lem:opt:high:prob}
Let $\beta \in (0,1)$ and let $\mathcal{C}$ be an $\alpha$-approximation for $G$. Define $c = 54/[\beta^2(1-q)^3]$.
If $\opt(G) \geq c\log n$,  then
with probability at least $1-n^{-2.5}$, $G' \sim \mathcal{S}_q(G)$ has
$\text{cost}_{G'}(\mathcal{C}) \leq \frac{18\alpha}{(1-\beta)(1-q)^2} \cdot \OPT(G').$
\end{restatable}

\ifarxiv
\begin{proof}
\noindent\textbf{Step 1: Lower bounding $\OPT(G')$ whp.}
By~\Cref{lem:opt:triangle:bound}, the maximum set of edge-disjoint bad triangles $T(G)$ of $G$ satisfies $|T(G)| = \tau(G) \geq \nicefrac{\opt(G)}{9}$.
Define $Y_t = \mathbf{1}[t \subseteq E']$, where $E'$ is the set of edges in $G'$. Since each edge $e \in t$ satisfies $e \in E'$ with an independent probability of $(1-q)$, we get $\ex[Y_t] = (1-q)^3$.
The $Y_t$'s for $t \in T(G)$ are mutually independent because the triangles in $T(G)$ are edge-disjoint. Define $Y = \sum_{t\in T} Y_t$. Then,
\begin{align*}
  \mu_Y := \ex[Y]& = (1-q)^3 \cdot \tau(G)
  \geq \nicefrac{(1-q)^3}{9}\cdot  \opt(G)
  \geq \nicefrac{c(1-q)^3}{9}\cdot \log n,
\end{align*}
where the penultimate inequality uses~\Cref{lem:opt:triangle:bound} and the last inequality uses the assumption that $\opt(G) \ge c \cdot \log n$.
Applying a Chernoff bound with $\delta = \beta$ yields
\begin{align*}
  \Pr[Y \leq (1-\beta)\mu_Y]
  &\leq \exp\!\left(-\nicefrac{\beta^2}{2}\cdot \mu_Y\right)\\
  &\leq \exp\!\left(-\nicefrac{\beta^2 c(1-q)^3}{18}\cdot \log n\right)\\
  &= n^{-\beta^2 c(1-q)^3/18}.
\end{align*}

By choice of $c = 54/[\beta^2(1-q)^3]$, we have $n^{-\beta^2 c(1-q)^3/18} = n^{-3}$. Thus, with probability at least $1-n^{-3}$,
\[
  \opt(G') \ge Y \geq (1-\beta)\mu_Y \geq \nicefrac{(1-\beta)(1-q)^3}{9} \cdot \opt(G),
\]
where the first inequality uses~\Cref{lem:opt:triangle:bound}.

\noindent\textbf{Step 2: Upper bounding $\text{cost}_{G'}(\mathcal{C})$ whp.}
Next, consider $\text{cost}_{G'}(\mathcal{C})$. Let $D$ denote the set of edges that incur disagreements in clustering $\mathcal{C}$ for graph $G$. That is, $\text{cost}_G(\mathcal{C}) = |D|$. Since $G'$ is a subgraph of $G$, we have $\text{cost}_{G'}(\mathcal{C}) = \sum_{e \in D} X_e$. Recall that $X_e$ is an indicator for the event that edge $e$ survives the deletion process. By definition of our model, we have $\ex[X_e] = 1-q$ and, thus, $\ex[\sum_{e \in D} X_e] = (1-q) |D|$. Define $X = \sum_{e \in D} X_e$ and let $\mu_X := \ex[X] = (1-q) |D|$.

Applying a Chernoff bound with $\delta = \beta \cdot (1-q)$ yields
\begin{align*}
    \Pr[X \ge (1+\beta(1-q))(1-q)|D|]
    &\le \exp\left(-\nicefrac{\beta^2(1-q)^2 }{3}\cdot \mu_X\right)
    = \exp\left(-\nicefrac{\beta^2(1-q)^3 }{3}\cdot |D|\right).
\end{align*}
Plugging in $c = 54/[\beta^2(1-q)^3]$ and $|D| \ge \opt(G) \ge c \log n$, and observing that $(1+\beta(1-q)) \le 2$, we get
\begin{align*}
\Pr[X \ge 2 \cdot (1-q) |D|]
&\le \Pr[X \ge (1+\beta(1-q))(1-q)|D|]\\
&\le \exp(-18 \log n) = n^{-18}.
\end{align*}
This implies that $\text{cost}_{G'}(\mathcal{C}) \le 2 \cdot (1-q) \cdot \text{cost}_G(\mathcal{C})$ holds with probability at least $1-n^{-18}$.

\noindent\textbf{Step 3: Combining both bounds.}
Applying a union bound to the event that $Y \le (1-\beta) \mu_Y$ or $X \ge 2 (1-q) |D|$ yields
\begin{align*}
    \Pr[Y \le (1-\beta) \mu_Y \lor X \ge 2 (1-q)|D|]
    &\le \Pr[Y \le (1-\beta) \mu_Y] + \Pr[X \ge 2(1-q)|D|]\\
    &\le n^{-3} + n^{-18}\\
    &\le n^{-2.5},
\end{align*}
where the last inequality assumes an instance with $n \ge 2$.
Hence, with probability at least $1-n^{-2.5}$, we have both
$\text{cost}_{G'}(\mathcal{C}) \le 2 \cdot (1-q) \cdot \text{cost}_{G}(\mathcal{C})$ and
 $\opt(G') \ge  \nicefrac{(1-\beta)(1-q)^3}{9} \cdot \opt(G)$.
Using that $\mathcal{C}$ is an $\alpha$-approximation  for $G$, we can conclude that the following chain of inequalities holds with probability at least $1-n^{-2.5}$,
\begin{align*}
     \text{cost}_{G'}(\mathcal{C}) \le (1-q) \cdot 2 \cdot \text{cost}_G(\mathcal{C})
     \le (1-q) \cdot 2 \cdot  \alpha \cdot \opt(G)
     \le \frac{18\alpha}{(1-\beta)(1-q)^2} \cdot \OPT(G').
\end{align*}
 \end{proof}
\fi

\subsection{Upper bounds with high probability}
\label{sec:min:whp:bounds}

In this section, we design algorithms that achieve improved approximation guarantees (compared to the general graph setting) with high probability. If the algorithm has access to both $G$ and $G'$, we prove the guarantee of~\Cref{thm:l1:both:graphs}, which we restate here for convenience.

\thmBothGraphs*

By~\Cref{lem:opt:high:prob}, every constant factor approximation for $G$ satisfies the theorem provided that $\opt(G)$ is sufficiently large. Hence, our strategy for proving~\Cref{thm:l1:both:graphs} is to design a dedicated algorithm for the case when $\opt(G)$ is small. We can then combine both cases by computing two clusterings, and returning the one with smaller cost. Notably, our algorithm for small $\opt(G)$ only needs $G'$ as input.

\begin{restatable}{theorem}{LowCostDirect}
\label{low-cost-direct}
    If $\OPT(G) < c \cdot \log^\gamma n$ for $c,\gamma \ge 1$, then there is a polynomial time algorithm that, given only $G' \sim \mathcal{S}_{q}(G)$, computes an $O(\log (\nicefrac{c^2}{(1-q)^2} \cdot  \log^{2\gamma} n))$-approximation for $G'$ with high probability. 
\end{restatable}

We proceed by sketching some arguments for proving~\Cref{low-cost-direct}, which directly implies~\Cref{our_result_low-cost-direct_}. Let $\gamma,c \ge 1$ be hyperparameters, and assume that $\opt(G) < c \cdot \log^\gamma n$. The basic idea of our algorithm is to exploit structural properties of $G'$ (implied by the assumption $\opt(G) < c \log^\gamma n$) to recover certain clusters of the optimal solution for $G'$. After removing the vertices in those clusters, the remaining vertices $R$ will induce a subgraph $G'[R]$ with $n_r \in O(c^2 \cdot \log^{2\gamma} n)$ vertices. This allows us to cluster the remaining vertices by running the $O(\log n)$-approximation for general graphs~\cite{demaine2006correlation} to achieve an $O(\log n_r)= O(\log(c^2 \cdot \log^{2\gamma} n))$-approximation.

The full algorithm is formalized in~\Cref{alg:low_cost}. We continue by discussing the individual steps. First, call a cluster $C$ \emph{trivial} if $C$ incurs no disagreements in $G'$. That is, there is no $e \in E'^-$ with $e \subseteq C$ and no $e \in E'^+$ with $|e \cap C| = 1$. It is not hard to show that we can identify all trivial clusters in polynomial time. As a preprocessing step, our algorithm finds and removes all such clusters. This can be done by computing the connected components of $G'[E'^+]$ and checking whether they cause disagreements.

\begin{algorithm}[t]
\SetAlgoLined
\KwIn{Graph $G' = (V, E'^{+}\cup E'^-)$, deletion probability $q \in [0,1)$, and parameters $\gamma,c \ge 1$}
\KwOut{Clustering of $V$}
\SetKwProg{Proc}{Procedure}{}{}

\BlankLine

\textbf{Remove trivial clusters} (using \Cref{lem:trivial:cluster})
\BlankLine
\textbf{Step 1: Build Similarity Graph $H$} \\

Initialize empty graph $H = (V, \emptyset)$\;
\For{each pair $(u,v) \in V \times V$}{
    \lIf{$|N_{G'}^+(u) \cap N_{G'}^+(v)| \ge 48 c  \cdot \log^\gamma n$}{
        Add edge $\{u,v\}$ to $H$
    }
}

\BlankLine
\textbf{Step 2: Extract and Filter Components} \\
Find connected components $S_1, S_2, \dots, S_k$ of $H$\;
$C_{\text{p}} \gets \emptyset$, $R \gets \emptyset$\;
\For{each component $S_i$}{
    \lIf{$|S_i| \geq \frac{100  c \cdot \log^\gamma n}{(1-q)^2}$}{$C_{\text{p}} \gets C_{\text{p}} \cup \{S_i\}$
    }
    \lElse{
        $R \gets R \cup S_i$
    }
}

\BlankLine
\textbf{Step 3: Cluster Residual Set} \\
Run general graph algorithm on $G'[R]$\;
Let $C_R$ be the resulting clustering\;

\BlankLine
\Return $C_{\text{p}} \cup C_R$\;
\caption{Algorithm for small $\opt(G)$}
\label{alg:low_cost}
\end{algorithm}

\begin{restatable}{lemma}{lemTrivialCluster} \label{lem:trivial:cluster}
    All trivial clusters can be identified in polynomial time.
\end{restatable}

\ifarxiv
\begin{proof}
    Construct the graph containing only the positive edges of $G'$. Any trivial
    cluster induces a connected component of this graph with no positive edge to
    the rest of the graph. Computing the connected components, after which each component can be verified by inspecting
    its incident edges. Hence, all trivial clusters can be identified and removed in
    polynomial time.
\end{proof}
\fi

The first and second step of our algorithm are motivated by the following auxiliary lemma, which exploits that $G'$ is sampled from $\mathcal{S}_q(G)$. 

\begin{restatable}{lemma}{lemPreserveLargeClusters}
    \label{lem:l1:preserve:large:clusters}
    Let $T_1,\ldots, T_\ell$ denote the clusters of size $|T_i| \ge \theta_{\mathrm{size}} := \frac{100  c \cdot \log^\gamma n}{(1-q)^2}$ in an optimal solution for $G$. For $G' \sim \mathcal{S}_q(G)$, with probability at least $1-n^{-3}$, the following two properties hold for \emph{all} $T_i$:
    \begin{enumerate}
        \item $\forall u,v \in T_i$ with $u \not= v\colon$ $|N^+_{G'}(u) \cap N^{+}_{G'}(v)| \geq 48 c \cdot  \log^\gamma n$
        \item $\forall v \in V \setminus T_i\colon$ $|N^-_{G'}(v) \cap T_i| \geq 48 c \cdot  \log^\gamma n$
    \end{enumerate}
\end{restatable}

\ifarxiv
\begin{proof}
    Our proof proceeds in three steps:
    \begin{enumerate}
        \item Fix vertices $u \not= v$ with $u,v \in T_i$ for some $i$. Let $S_{u,v}$ denote the event that $|N^+_{G'}(u) \cap N^{+}_{G'}(v)| < 48  c \cdot \log^\gamma n$. Below, we show that
        $
        \Pr[S_{u,v}] \le n^{-12}.
        $
        \item For a vertex $v$, define $Y_v$ to be the event that there exists a $T_i$ with $v \not\in T_i$ and $|N^-_{G'}(v) \cap T_i| <  48  c \cdot \log^\gamma n$.
        Below, we show that
        $
        \Pr[Y_{v}] \le n^{-11}.
        $
        \item Once, we have established the inequalities of the first two steps, we are ready to prove the statement.
        Let $Z$ denote the event that $S_{u,v} = 1$ for at least one pair $u \not= v$ with $u,v \in T_i$ for some $i$ or $Y_v=1$ for at least one $v$.
        Using a union bound and the inequalities of the first two steps yields
        \begin{align*}
            \Pr[Z] &\le \sum_{i \in [\ell]} \sum_{u,v \in T_i \colon u \not= v} \Pr[S_{u,v}] + \sum_{v \in V} \Pr[Y_v]\\
            &\le \sum_{i \in [\ell]} \sum_{u,v \in T_i \colon u \not= v} n^{-12} + \sum_{v \in V} n^{-11}.\\
            &
            \le n^{-3},
        \end{align*}
        where the second inequality uses the bounds of the first two steps.
        The probability that the properties of the lemma hold is $1-\Pr[Z] \ge 1 - n^{-3}$.
    \end{enumerate}

    \paragraph{Proof of Step 1.}
        Fix two vertices $u \not= v$ with $u,v \in T_i$ for some $i \in [\ell]$.
        Each $w \in T_i \setminus \{u,v\}$ is a common positive neighbor of $u$ and $v$ unless
        $uw$ or $vw$ is a disagreement. By assumption that $\opt(G) < c  \log^\gamma n$,
        at most $2 \cdot \opt(G) < 2c  \log^\gamma n$ such exceptions can exist.
        Hence:
        \begin{align*}
            |N^{+}_G(u) \cap N^{+}_G(v)|
            &\geq |T_i| - 2 - 2 \opt(G)\\
            &\geq \theta_{\mathrm{size}} - 2 - 2c\log^\gamma n\\
            &\geq \frac{97c\log^\gamma n}{(1-q)^2}. \tag{A1}\label{eq:A1}
        \end{align*}

          For each $z \in N^{+}_G(u) \cap N^{+}_G(v)$, define $W_z = X_{vz} \cdot X_{zu}$. Recall that $X_{vz}$ is an indicator variable for the event that $vz$ survives the deletion process. Then, $\ex[W_z] = (1-q)^2$, as edge deletions are independent.

        The $W_z$'s are mutually independent as they concern disjoint edge pairs.
        Let $W = \sum_{z} W_z = |N^{+}_{G'}(u) \cap N^{+}_{G'}(v))|$.
        By \eqref{eq:A1},
        \[
            \mu_W := E[W] = (1-q)^2\cdot |N_G^{+}(u)\cap N_G^{+}(v)| \geq 97c\cdot \log^\gamma n.
        \]
        Applying the multiplicative Chernoff bound with
        $\delta=\nicefrac12$ yields
        \begin{align*}
            \Pr[S_{u,v}] &\le \Pr\!\Bigl[W\le\nicefrac12\cdot \mu_W \Bigr]
             \le \exp\Bigl(-\nicefrac{\mu_W}{8}\Bigr)\\
             &\le \exp\Bigl(-\nicefrac{97c}{8}\cdot \log^{\gamma}n\Bigr)\\
             &\le n^{-97/8} \le n^{-12},
        \end{align*}
        where the penultimate step uses $\gamma\ge 1$ and $c \ge 1$, so $c \cdot \log^{\gamma}n\ge\log n$ and
        $\exp(-\nicefrac{97c}{8}\cdot \log^{\gamma}n)\le\exp(-\nicefrac{97}{8}\cdot \log n)=n^{-97/8}$.

        \paragraph{Proof of Step 2.} Fix a vertex $v$ and some $T_i$ with $v \not\in T_i$. Observe that $|N^-_G(v) \cap T_i| \ge \nicefrac{99c}{1-q} \cdot  \log^\gamma n$, since otherwise the completeness of $G$ would imply $|N^+_G(v) \cap T_i| > c \log^{\gamma} n$, which is a contradiction to $|N^+_G(v) \cap T_i|  \le \opt(G) < c \log^{\gamma} n$.

        For each $u \in N^-_G(v) \cap T_i$, define $W_u =  X_{uv}$. Then, $\ex[W_u] = 1-q$. Let $W = \sum_{u \in N^-_G(v) \cap T_i} W_u$. Then,
        $$
            \mu_W := \ex[W] = (1-q) \cdot |N^-_G(v) \cap T_i| \ge 99 \cdot c \log^\gamma n.
        $$
        The $W_u$'s are mutually independent as they concern disjoint edges.
        Applying the multiplicative Chernoff bound with
        $\delta=\tfrac12$ yields
        \begin{align*}
            \Pr\Bigl[W\le\nicefrac12 \cdot \mu_W \Bigr]
            &\le \exp\Bigl(-\nicefrac{1}{8}\cdot \mu_W\Bigr)\\
            &\le \exp\Bigl(-\nicefrac{99c}{8}\cdot \log^{\gamma}n\Bigr)\\
            &\le n^{-\nicefrac{99}{8}} \le n^{-12},
        \end{align*}
        where the penultimate step uses $\gamma\ge 1$ and $c \ge 1$, so $c \log^{\gamma}n\ge\log n$ and
        $\exp(-\nicefrac{
        99c}{8}\cdot \log^{\gamma}n)\le\exp(-\nicefrac{99}{8}\cdot \log n)=n^{-\nicefrac{99}{8}}$.

        Let $Y^i_v$ denote the event that $|N^-_{G'}(v) \cap T_i| < 48 c\cdot  \log^\gamma n$. Clearly, $\Pr[Y^i_v] \le \Pr\!\Bigl[\,W\le\nicefrac12 \cdot \mu_W\,\Bigr] \le n^{-12}$. Applying a union bound yields
        $$
        \Pr[Y_v] \le \sum_{i \in [\ell]} \Pr[Y^i_v] \le n \cdot n^{-12} = n^{-11}.
        $$
   \end{proof}
\fi

The main implication of~\Cref{lem:l1:preserve:large:clusters} is that all large clusters in an optimal solution for $G$ are also part of an optimal solution for $G'$ with high probability.

\begin{corollary}
    \label{cor:preserve:large:clusters}
    For $G' \sim \mathcal{S}_q(G)$, with probability at least $1-n^{-3}$, every optimal clustering for $G'$ contains $T_1,\ldots, T_\ell$ of~\Cref{lem:l1:preserve:large:clusters}.
\end{corollary}

\begin{proof}
    With probability at least $1-n^{-3}$ the properties of~\Cref{lem:l1:preserve:large:clusters} hold for $G'$. Let $\mathcal{C}^*$ denote an optimal clustering for $G'$ conditioned on the properties of~\Cref{lem:l1:preserve:large:clusters} holding.

    Fix some $T_i$. We can observe that all pairs $u,v \in T_i$ have to be in the same cluster of $\mathcal{C}^*$, as if $u$ and $v$ were separated, then each element of $N^+_{G'}(u) \cap N^{+}_{G'}(v)$ incurs a cost. Using the first property of~\Cref{lem:l1:preserve:large:clusters}, this gives $\text{cost}_{G'}(\mathcal{C}^*) \ge 48 c \cdot \log^\gamma n > \opt(G) \ge \opt(G')$; a contradiction.

   Hence, there is some $C \in \mathcal{C}^*$ with $T_i \subseteq C$. It remains to show $C \subseteq T_i$. Assume there is a $v \in C \setminus T_i$. By the second property of~\Cref{lem:l1:preserve:large:clusters}, the vertex $v$ incurs a cost of at least $48 c \log^\gamma n > \opt(G) \ge \opt(G')$; a contradiction.
\end{proof}

Step 1 and Step 2 of~\Cref{alg:low_cost} are designed to recover the large clusters $T_1,\ldots, T_\ell$ provided that the properties of~\Cref{lem:l1:preserve:large:clusters} hold.

\begin{restatable}{observation}{ObsLargeClusters}
    \label{obs:find:large:clusters}
    If the two properties of~\Cref{lem:l1:preserve:large:clusters} hold, then $C_{\text{p}}$ of~\Cref{alg:low_cost} contains exactly the clusters $T_1,\ldots, T_\ell$.
\end{restatable}

\ifarxiv
\begin{proof}
    Fix some $T_i$.
    By the first property of~\Cref{lem:l1:preserve:large:clusters}, we have $T_i \subseteq H'$ where $H'$ is some connected component of the similarity graph of~\Cref{alg:low_cost}. If $H'\setminus T_i \not= \emptyset$, then there are vertices $v \in H'\setminus T_i$ and $u \in T_i$ with $$|N_{G}^+(v) \cap N_{G}^+(u)| \ge |N_{G'}^+(v) \cap N_{G'}^+(u)| \ge 48 c \cdot \log^\gamma n.$$
    However, since $T_i$ is an optimal cluster for $G$ and $v \not\in T_i$, this implies $\opt(G) > c \log^\gamma n$, which is a contradiction. Conversely, let $S$ be a connected component of $H$ with
$|S| \ge \theta_{\mathrm{size}}$. Every edge $uv$ of $H$ satisfies
$|N^{+}_{G}(u) \cap N^{+}_{G}(v)| \ge |N^{+}_{G'}(u) \cap N^{+}_{G'}(v)|
\ge 48c \cdot \log^{\gamma} n$, so by the argument above $u$ and $v$ lie in the same cluster of the optimal
solution for $G$. Since any two vertices of $S$ are joined by a path in $H$,
induction along the path shows that all vertices of $S$ lie in a common
cluster $C$ of that solution. Then $|C| \ge |S| \ge \theta_{\mathrm{size}}$,
so $C = T_i$ for some $i$, and $T_i$ is itself a
component of $H$. As $S \subseteq T_i$ and components are disjoint, $S = T_i$.
\end{proof}
\fi

\Cref{obs:find:large:clusters} states that Step 1 and 2 recover the large clusters in the optimal solution of $G$. After these two steps, the algorithm clusters the remaining vertices using the algorithm of~\citet{demaine2006correlation}. One aspect of our proof is to observe that the number of such vertices is small.

\begin{restatable}{observation}{ObsFewRemain}
    \label{obs:few:remaining:vertices}
    If the two properties of~\Cref{lem:l1:preserve:large:clusters} hold, then $|R| \le  O\left (\nicefrac{c^2}{(1-q)^2} \right ) \cdot  \log^{2\gamma} n$ with $R$ defined as in~\Cref{alg:low_cost}.
\end{restatable}

\ifarxiv
\begin{proof}
    In the optimal solution for $G$, each $v \in R$ is part of a non-trivial cluster with size $\le \theta_{\text{size}}$. By our assumption that $\opt(G) < c \log^{\gamma} n$, there can be at most $2 \cdot c \log^{\gamma} n$ such clusters. Hence, $|R| \le  \theta_{\text{size}} \cdot 2 \cdot c \log^{\gamma} n = O\left (\nicefrac{c^2}{(1-q)^2}\right )  \cdot \log^{2\gamma} n$.
\end{proof}
\fi

We can now combine the previous lemmas to prove the approximation factor claimed in~\Cref{low-cost-direct}.

\begin{proof}[Proof of \Cref{low-cost-direct}]
    By \Cref{lem:trivial:cluster}, we may assume for the rest of the proof that there are no trivial clusters.

    Let $T_1,\ldots, T_\ell$ denote the large clusters with $|T_i| \ge \theta_{\mathrm{size}}$ in an optimal solution for the complete graph $G$. With probability at least $1-n^{-3}$, these clusters are part of every optimal solution for $G'$ by~\Cref{cor:preserve:large:clusters}. By~\Cref{obs:find:large:clusters}, $T_i \in C_{\text{p}}$ for each $i \in [\ell]$. The cost incurred by these large clusters, i.e., the number of disagreements that contain at least one vertex in $\bigcup_i T_i$, is the same in the output of \Cref{alg:low_cost} and in the optimal solution.

    It remains to consider the cost incurred by the clusters in $C_R$. Let $c_R^*$ and $c_R$ denote the number of disagreements between vertices in $R$ in the optimal solution for $G'$ and in the algorithm's solution, respectively.
    As argued above, the ratio between $c_R$ and $c_R^*$ upper bounds the approximation ratio of our algorithm. 
    Since we are using the algorithm of~\citet{demaine2006correlation}, we get
    $
        c_R \le O(\log |R|) \cdot c_{R}^*.
    $
    By~\Cref{obs:few:remaining:vertices}, 
    $$
    c_R \le O (\log |R|) \cdot c_{R}^* \le O \left(\log (\nicefrac{c^2}{(1-q)^2} \cdot  \log^{2\gamma} n) \right) \cdot c_{R}^*.
    $$
\end{proof}

\ifarxiv
\begin{proof}[Proof of running time]
The algorithm first computes \( |N_{G'}^{+}(u)\cap N_{G'}^{+}(v)| \) for every pair of vertices \(u,v\in V\). These values can be computed in \(O(n^{\omega})\) time, where \(\omega<2.373\) is the exponent of matrix multiplication. Using these values, the auxiliary graph \(H\) is constructed by considering every pair of vertices, which requires \(O(n^2)\) time. The connected components of \(H\) are then computed in linear time \(O(|V(H)|+|E(H)|)\). Finally, we use the algorithm of~\citet{demaine2006correlation} on the residual graph. Since this algorithm runs in polynomial time, every step of Algorithm~\ref{alg:low_cost} runs in polynomial time. Therefore, the overall running time is polynomial.
\end{proof}
\fi

Using~\Cref{low-cost-direct} and Lemma \ref{lem:opt:high:prob}, we can prove~\Cref{thm:l1:both:graphs} by picking the appropriate values for $\gamma$ and $c$.

\ifarxiv
\thmBothGraphs*

\begin{proof}
    Let $\beta =\frac{1}{2}$, $c = 54/[\beta^2(1-q)^3]$, and $\gamma = 1$. Consider the following algorithm:
    \begin{enumerate}
        \item Compute a clustering $\mathcal{C}_1$ for $G$ by running some constant $\alpha$-approximation for the complete graph setting on $G$.
        \item Compute a clustering $\mathcal{C}_2$ for $G'$ by running~\Cref{alg:low_cost} with parameters $\gamma$ and $c$ on $G'$
        \item Let $\mathcal{C}$ denote whichever clustering, $\mathcal{C}_1$ or $\mathcal{C}_2$, has smaller cost on $G'$; return $\mathcal{C}$.
    \end{enumerate}

    We distinguish the two cases (1) $\opt(G) \ge c \log n$ and (2) $\opt(G) < c \log n$.

    \paragraph{Case (1):} Assume $\opt(G) \ge c \log n$. By~\Cref{lem:opt:high:prob}, we have
    \begin{align*}
    \text{cost}_{G'}(\mathcal{C}) \leq \frac{18\alpha}{(1-\beta)(1-q)^2} \cdot \OPT(G')
    \leq \frac{36\alpha}{(1-q)^2} \cdot \OPT(G')
    \leq O(\nicefrac{1}{(1-q)^2})  \cdot \OPT(G')
    \end{align*}
    with probability at least $1-n^{-2.5}$.

   \paragraph{Case (2):} Assume $\opt(G) < c \log n$. By~\Cref{low-cost-direct}, we have
   \begin{align*}
       \text{cost}_{G'}(\mathcal{C}) &\le O \left (\log \left (\nicefrac{c^2 }{(1-q)^2} \cdot \log^{2\gamma} n\right )\right) \cdot \opt(G')\\
       &\le
O(\log\left(\nicefrac{1}{(1-q)^8}\cdot  (\log n)^2\right) \cdot \opt(G')\\
        &= O(\log\log n + \log(\nicefrac{1}{1-q}))
   \end{align*}
    with probability at least $1-n^{-3}$.

    In both cases, with probability at least $1-n^{-2.5}$ we have that
    \begin{align*}
         &\text{cost}_{G'}(\mathcal{C})
         \le O(\max\{\nicefrac{1}{(1-q)^2}, \log\log n + \log(\nicefrac{1}{1-q})\}) \cdot \opt(G').
    \end{align*}
\end{proof}
\fi

\section{Min-max Objective}

In this section, we will show that an algorithm given only $G'$ achieves a cost of $O(1/(1-q)) \cdot \max\{c\log n, \mathrm{OPT}_{\infty}(G')\}\footnote{If $\OPT_\infty(G')=0$, the connected components of $G'[E'^+]$ are a
zero-cost clustering. Hence we can assume $\OPT_\infty(G') \ge 1$ }$ with high probability, in particular:

\begin{restatable}{theorem}{thmMinMaxCorrelation}
\label{lemma:min_max_correlation}
    Given $G' \sim \mathcal{S}_{q}(G)$, there exists an algorithm that, with high probability, returns a clustering for $G'$ with cost at most $ O(1/(1-q)) \cdot \max\{\log n, \mathrm{OPT}_{\infty}(G')\}$. 
\end{restatable}

Throughout this section, we let $\text{cost}_{G'}(\mathcal{C}) = \max_{v \in V} \text{cost}_{G'}(\mathcal{C},v),$ so the optimal cost is denoted by
$
    \text{OPT}_{\infty}(G') = \min_{\mathcal{C}} \text{cost}_{G'}(\mathcal{C}).$

We begin by providing a key lemma from the $4$-approximation algorithm for the min-max objective on complete graphs given by \citet{heidrich2024}.
\begin{lemma} \cite{heidrich2024}
\label{lem:og-4apx}
 For $G$ complete, fix $t \ge \text{OPT}(G)$.  If $|N^+(u) \cap N^+(v)| > 2t$, then $u$ and $v$ must belong to the same cluster in any optimal clustering.
If $|N^+(u) \Delta N^+(v)| > 2t$, then $u$ and $v$ must belong to different clusters in any optimal clustering.
\end{lemma}
Although Lemma \ref{lem:og-4apx} holds for complete graphs, it inspires Algorithm \ref{alg:min-max} and informs the algorithm's analysis.

\paragraph{Algorithm.} We note that Algorithm \ref{alg:min-max} is only given the incomplete graph $G'$.
We choose $c$ as a large constant. 
Since $\opt_{\infty}(G')$ is not known in advance, we run
\Cref{alg:min-max} for every guess
$d \in \{c\log n,\, 2c\log n,\, 4c\log n,\, \ldots,\, n\}$ \raj{Should we be more explicit about this being binary search}\jens{It's doubling and not binary search, right? Up to you, add the keyword if you like, but I think it is clear enough} which returns clustering $\mathcal{C}_d$, then we compute
$\text{cost}_{G'}(\mathcal{C}_d)$ for each $\mathcal{C}_d$, and output the partition of minimum cost.
    

\begin{algorithm}[h]
\caption{\small Robust Min-Max Correlation Clustering }
\label{alg:min-max}
\KwIn{ $G'=(V,E'^+ \cup E'^-)$, deletion probability $q\in[0,1)$, guess $d$, parameter $\lambda >0$}
\KwOut{Partition $\mathcal{C}$}

$S \gets \{v\in V : |N^{+}_{G'}(v)|>\lambda d\}$ 

Initialize auxiliary graph $H=(V,\emptyset)$\;

\ForEach{$v\in S$}{
    \ForEach{$u\in V,\;u\neq v$}{
        \If{$|N_{G'}^{+}(u)\cap N_{G'}^{+}(v)|>2d$}{
            add edge $(u,v)$ to $H$\;
        }
    }
}

$\mathcal{C} \gets$ connected components of $H$\;


\Return{$\mathcal{C}$}

\end{algorithm}





By slightly adapting the proof of Lemma  \ref{lem:og-4apx}, the following two lemmas combined give the analogous result for non-complete graphs.
\begin{restatable}{lemma}{lemForced}
\label{lem:forced}
Let $G'$ be a signed (not necessarily complete) graph, and let $\mathcal{C}$ be a partition of $G'$
such that $\text{cost}_{G'}(\mathcal{C})\le d$. If $u,v \in V$ have
$|N^{+}_{G'}(u)\cap N^{+}_{G'}(v)|>2d$, then
$[u]_{\mathcal{C}}=[v]_{\mathcal{C}}.$
\end{restatable}

\begin{restatable}{lemma}{lemCrossneg}
\label{lem:crossneg}
Let $G'$ be a signed (not necessarily complete) graph and let $\mathcal{C}$ be a partition of $G'$
such that $\text{cost}_{G'}(\mathcal{C})\le d$. If $u,v \in V$ have
$|N^{+}_{G'}(v) \cap N^{-}_{G'}(u)|>2d$, then $
[u]_{\mathcal{C}}\neq [v]_{\mathcal{C}} .$
\end{restatable}

Next, Lemma \ref{lem:event_low_overlap} proves that for
a high-degree vertex $v$, with high probability, every vertex $u$ falls into one of the two
cases captured by Lemmas \ref{lem:forced} and \ref{lem:crossneg}. In other words, the deletion of edges is very unlikely to destroy so many edges between $u$ and
$N^+_{G'}(v)$ that the choice on whether to cluster $u$ with $v$ is ambiguous. The proof follows from a Chernoff bound.


\begin{definition}
\label{def:low_overlap_event}
Fix $\lambda > 0$ and $d > 0$. Define
$
\mathcal{E}_{\lambda,d}$ to be the event that for all $v \in V$ with $ |N_{G'}^{+}(v)| > \lambda d$, every $ u \in V \setminus \{v\} $ has $ |N_{G'}^{+}(v) \cap N_{G'}^{+}(u)| + |N_{G'}^{+}(v) \cap N_{G'}^{-}(u)| > 4d.
$
\end{definition}

\begin{restatable}{lemma}{lemOverlap}
\label{lem:event_low_overlap}
Fix $\lambda > \frac{8}{1-q}$ and  $d \geq  c \log n$
for $c \geq \frac{24}{(1-q)\lambda}$.
Then $\Pr[\mathcal{E}_{\lambda,d}] \geq 1 - \frac{1}{n}$.
\end{restatable}

The rest of the vertices are made singletons, and we show their cost is at most $\lambda d$. 

\begin{restatable}{lemma}{lemSingleton}
\label{lem:singleton}
Fix $\lambda>1$ and $d>0$ with $\OPT_\infty(G') \leq d$.
Condition on $\mathcal{E}_{\lambda,d}$, and let $\mathcal{C}$ be the
clustering returned by \Cref{alg:min-max} on $G'$. Then every vertex $v$ that is
a singleton in $\mathcal{C}$ satisfies
$\mathrm{cost}_{G'}(\mathcal{C},v) \leq \lambda d$.
\end{restatable}



\begin{proof}[Proof of \Cref{lemma:min_max_correlation}]
Let $d = \mathrm{OPT}_{\infty}(G')$, $\lambda = \frac{8}{1-q} +1$, and $\hat{d} = \max\{c\log n, d\}$. 
Condition on $G'$ satisfying event $\mathcal{E}_{\lambda,\hat{d}}$ (which by Lemma \ref{lem:event_low_overlap} occurs with probability $1-n^{-1}$).
Run
Algorithm \ref{alg:min-max} with parameters $\hat{d}$ and $\lambda$, and let its output clustering be $\mathcal{C}$.
Note that this is not the output clustering of the full algorithm within the guessing framework, which will lose a factor of $2$ compared to $\mathcal{C}$ due to the doubling strategy.
Let $S = \{v \in V : |N_{G'}^{+}(v)| > \lambda\hat{d}\}$. 
Fix an optimal clustering $\mathcal{C}^*$ for $G'$.

\medskip
\noindent\textbf{Step 1: $H$ recovers $\mathcal{C}^*$.}
We conditioned on $\mathcal{E}_{\lambda,d}$, and so
for all $v \in S$ and $u \in V$, exactly one of the following holds:
\begin{itemize}
\item $|N_{G'}^{+}(v) \cap N_{G'}^{+}(u)| > 2\hat{d}$:
      by  ~\Cref{lem:forced} applied to $\mathcal{C}^*$ with $\text{cost}_{G'}(\mathcal{C}^*) = d \leq \hat{d}$,
      we have $[v]_{\mathcal{C}^*} = [u]_{\mathcal{C}^*}$, and the algorithm adds edge $(u,v)$ to $H$.
\item $|N_{G'}^{+}(v) \cap N_{G'}^{-}(u)| > 2\hat{d}$:
      by  ~\Cref{lem:crossneg} applied to $\mathcal{C}^*$,
      we have $[v]_{\mathcal{C}^*} \neq [u]_{\mathcal{C}^*}$, and no edge is added.
\end{itemize}
Hence the algorithm adds $(u,v)$ to $H$ if and only if $[v]_{\mathcal{C}^*} = [u]_{\mathcal{C}^*}$.
Since the inner loop ranges over all $u \in V$, every optimal cluster-mate of $v$
— including those in $V \setminus S$ — is connected to $v$ in $H$.
Therefore each connected component of $H$ containing at least one vertex $ v \in S$ coincides with exactly one cluster of $\mathcal{C}^*$,
and $[v]_{\mathcal{C}} = [v]_{\mathcal{C}^*}$.
Since all non-singleton components of $H$ contain at least one vertex in $S$ by construction, every non-singleton component in $H$ coincides with exactly one cluster of $\mathcal{C}^*$.

\medskip
\noindent\textbf{Step 2: Bounding the disagreement.}
 For every non-singleton vertex $v$, $\text{cost}_{G'}(\mathcal{C},v) = \text{cost}_{G'}(\mathcal{C}^*,v) \leq d \leq \lambda\hat{d}$.
For every singleton vertex $v$,  ~\Cref{lem:singleton} gives $\text{cost}_{G'}(\mathcal{C},v) \leq \lambda\hat{d}$.
Therefore \[ \text{cost}_{G'}(\mathcal{C}) \leq  O(1/(1-q)) \cdot \max\{c\log n, \mathrm{OPT}_{\infty}(G')\}.\]
\end{proof}
\paragraph{Improved time.} In the supplementary materials, we prove that if we allow extra randomization, the running time of Algorithm \ref{alg:min-max} can be made $\widetilde{O}(n^2)$, by using the approach of \citet{cao2025min} that is based on \cite{johnson1984extensions}.



\section{Experiments } \label{sec:experiments}



Our code is written in Python $3.11.8$ \ifarxiv\footnote{Code available at \url{https://github.com/aou219/MinMax_Correlation_Clustering_Experiments}.}\fi{}. We ran our experiments on Apple M2, with 8 cores and 8 GB of memory.  
For our experiments, we used the Stanford Large Network Dataset Collection~\cite{leskovec2012learning}. Specifically, we used the \emph{ego-Facebook} dataset containing $10$ graphs that are subgraphs of a social network from Facebook (throughout this section, we call these ego-graphs). Each subgraph represents a specific user's friend list and the connection within. The graph is converted into a complete signed graph by representing friends as a positive edge and non-friends as a negative edge. We then subsample this complete graph with different sampling probabilities $q$ to generate instances $G' \sim \mathcal{S}_q(G)$. 
Graph statistics for the ego-graphs are in Table \ref{tab:stats} in the supplementary materials. We overview our main findings: 
\begin{itemize}
    \item \emph{(Min-max)}: On instances drawn from smaller Facebook ego-graphs, \Cref{alg:min-max} obtains an approximation that is at least as good as our guarantee from Theorem \ref{lemma:min_max_correlation}. Moreover, the approximation factor gets worse as $q$ increases, as expected. See the right side of Figure \ref{fig:minmax:results}.
    \item \emph{(Min-max)}: On instances drawn from the larger ego-graphs, it is computationally infeasible to obtain a lower bound on the true optimal cost. Instead we compare the output of \Cref{alg:min-max} against a surrogate for the optimal min-max cost, and still find guarantees against this proxy consistent with our theoretical results. See the left side of Figure \ref{fig:minmax:results}.  
    \item \emph{(Min-disagreement)}: In Lemma \ref{lem:sparsification_G}, we showed that any constant-factor approximate clustering for the min-disagreement objective for the complete graph $G$ is also a constant-factor approximation on $G' \sim \mathcal{S}_q(G)$ in expectation.
    While the theoretical result requires the clustering to be computed on $G$, our  experiments indicate that similar guarantees might hold when it is computed on $G'$: we run the \textsc{Pivot}~\citep{ACN-pivot} algorithm on planted-clique instances and on instances drawn from smaller Facebook ego-graphs and observe that the approximation remains constant.

     \item \emph{(Running time)}: One advantage of our Min-max algorithm is that it combinatorial, meaning that unlike the state-of-the-art it does not solve an LP. Thus it can be run on large graphs faster; see Figure \ref{fig:min_max_runtime} in the supplementary materials for running times.
\end{itemize}
We detail these highlighted results, 
and note that an extensive description of the experimental setup and more empirical results can be found in the supplementary materials.
\begin{figure*}[t]
    \centering
    \begin{minipage}{0.49\textwidth}
        \includegraphics[scale=0.5]{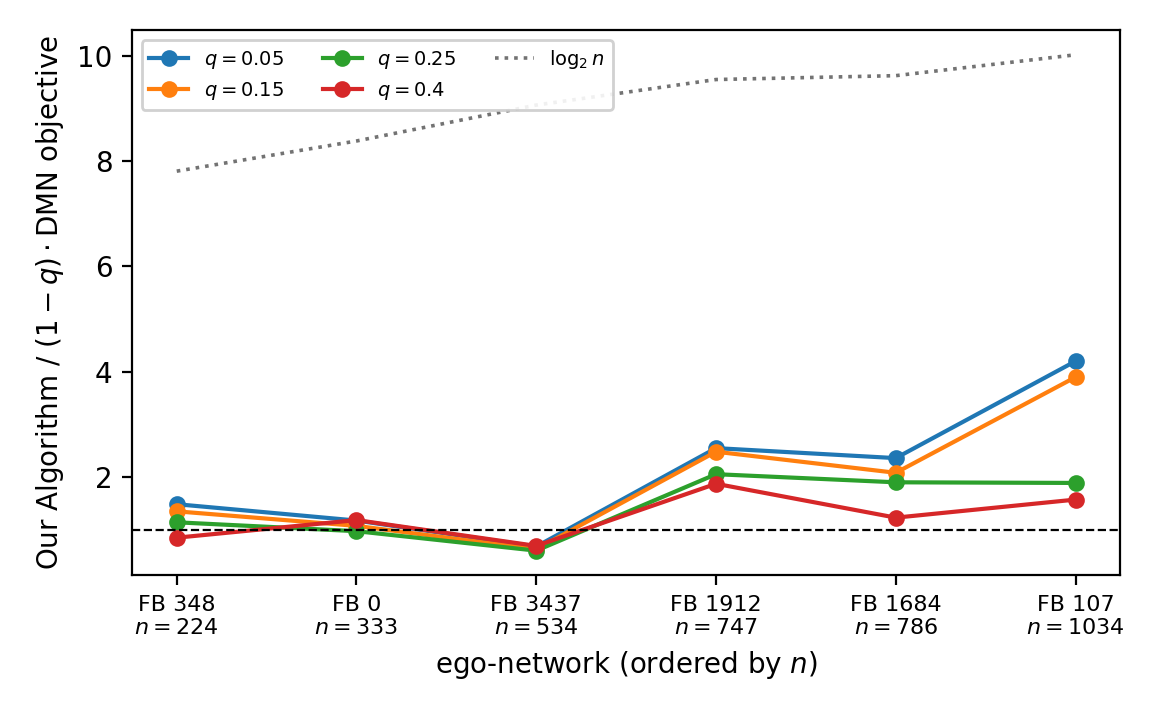 }
    \end{minipage}
    \begin{minipage}{0.49\textwidth}
        \includegraphics[scale=0.5]{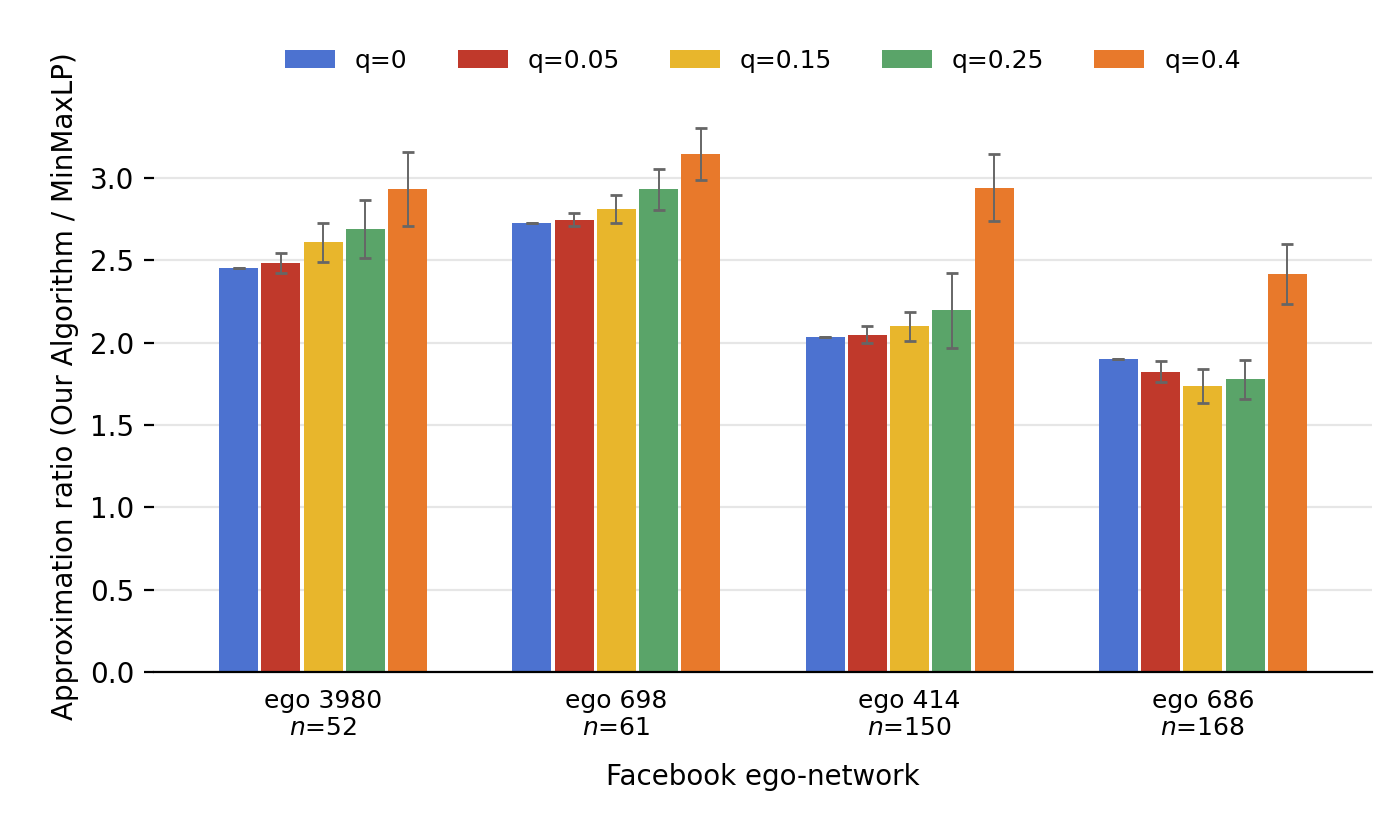}
    \end{minipage}
    \caption {Performance of \Cref{alg:min-max} on the Facebook ego-graphs under random
edge deletion, as means over $30$ seeds. \textbf{Left:} the six largest
ego-graphs, showing the cost on $G'$ relative to $(1-q)$ times the cost
reported by~\citet{DMN23} on the undeleted graph; one line per deletion
probability $q$, with $\log_2 n$ dotted for reference. \textbf{Right:} the
four smallest ego-graphs, where the min-max LP is solved on $G'$, so the ratio is an
upper bound on the true approximation ratio.}
    \label{fig:minmax:results}
\end{figure*}

\paragraph{Min-max.} For each of the $10$ Facebook ego-graphs and each deletion probability
$q \in \{0.05, 0.15, 0.25, 0.4\}$, we generated $30$ realizations $G'$ by
deleting each edge independently with probability $q$, and ran
\Cref{alg:min-max} on each. We compare the resulting cost against two
different reference values, depending on the size of the graph. On the four
smallest ego-graphs we solve the min-max LP relaxation  with
Gurobi~13.0.2, obtaining a lower bound on the cost of an optimal solution, so the reported
quantity is a upper bound on the true approximation ratio. On the six larger ego-graphs, solving the
LP is computationally not feasible, and
instead we compare against the objective achieved by the combinatorial
algorithm of~\citet{DMN23}\footnote{To the best of our knowledge, this is the lowest maximum
disagreement reported for these instances.} on the original complete graph (hereafter called the \emph{DMN objective}), rescaled by
$(1-q)$ to account for the deleted edges. The latter is a comparison
against another algorithm's
cost, not a lower bound on the optimum.

\Cref{fig:minmax:results} shows results of our experiments. On the left side, we see the average approximation ratio of~\Cref{alg:min-max} for the six
larger ego-graphs. We have one plot for each deletion probability as well as one plot for $\log(n)$ as a reference to the algorithm's theoretical guarantee of $O(1/(1-q)) \cdot \max\{c\log n, \mathrm{OPT}_{\infty}(G')\}$ (cf~\Cref{lemma:min_max_correlation}). On these graphs, if we consider the rescaled DMN objective as a proxy for the optimal, then our algorithm performs at least as well as its theoretical guarantee.

On the right side of \Cref{fig:minmax:results}, we see the average approximation ratio of \Cref{alg:min-max} on the four smallest ego-graphs. Here we solve the min-max LP
exactly, so the cost of \Cref{alg:min-max} is compared against a lower bound
on the optimum. The ratios seem to increase with $q,$ which we view as the expected behavior. Note that this is \emph{not} the case for the plot on the left side, where the ratio of the largest deletion probability is smallest. We suspect this to be an artifact of comparing against the scaled DMN objective instead of the LP. 

\paragraph{Min-disagreement.}

\begin{figure*}[t]
    \centering
    \begin{minipage}{0.49\textwidth}
        \centering
        \includegraphics[width=\linewidth]{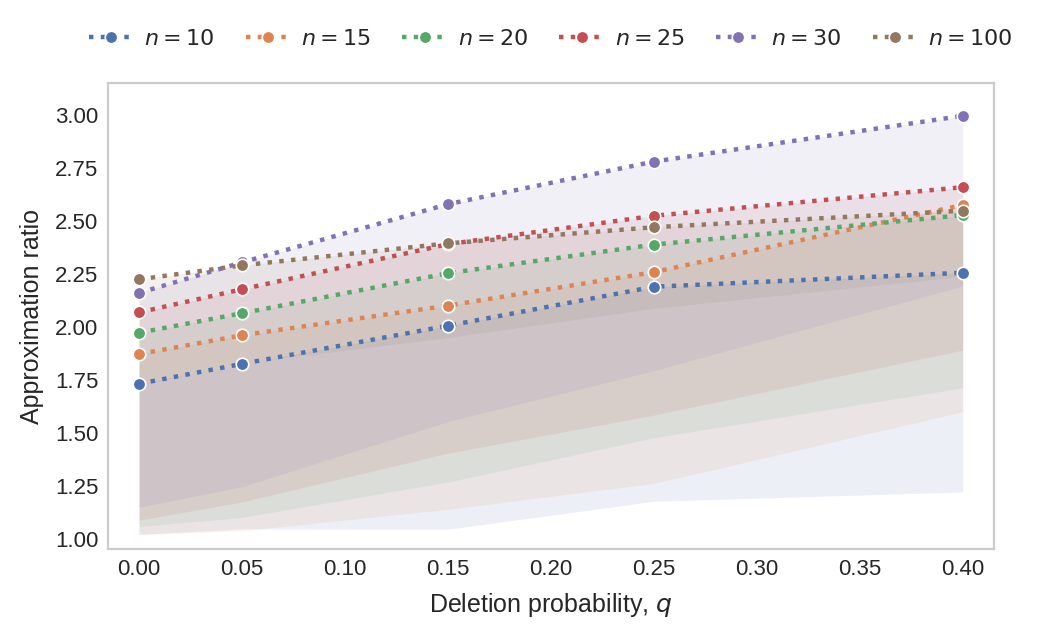}
    \end{minipage}
    \hfill
    \begin{minipage}{0.49\textwidth}
        \centering
        \includegraphics[width=\linewidth]{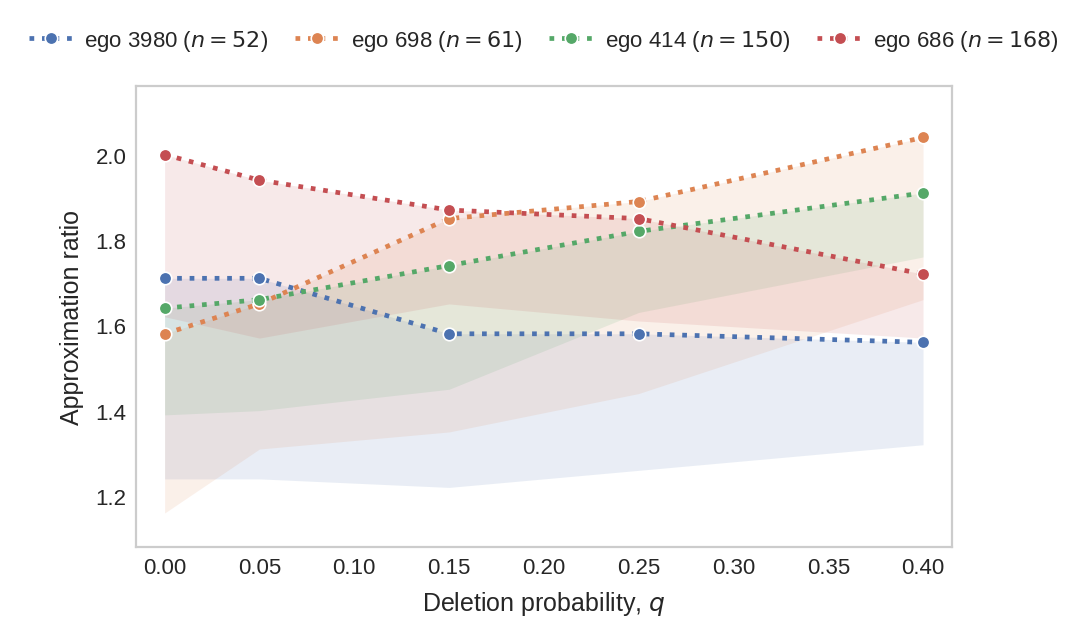}
    \end{minipage}
    \caption{Approximation ratio of \textsc{Pivot} against the LP optimum, as a function of the deletion probability $q$. Dotted lines show the mean ratio and the shaded region spans the
    best ratio observed to the mean. \textbf{Left:} planted-clique
    instances, over $50$ independent seeds per instance ($20$ seeds for
    $n=100$). \textbf{Right:} the four smallest Facebook ego-graphs, over
    $30$ seeds.}
    \label{fig:pivot:results}
\end{figure*}

For the min-disagreement objective we evaluate
\textsc{Pivot} run on non-complete graphs.
\Cref{lem:sparsification_G} bounds \textsc{Pivot} on the original complete
graph $G$, which an algorithm does not observe; we test whether the same
behavior holds on $G'$ itself. We use two families: planted-clique instances on $n \in \{10, 15, 20, 25, 30, 100\}$ vertices,
partitioned into cliques according to several decompositions per size, with
intra-clique density $p_{\mathrm{in}} = 0.9$ and inter-clique density
$p_{\mathrm{out}} = 0.1$ (the decompositions are listed in Table \ref{tab:clique_decompositions} in the supplementary
material), and the four smallest Facebook ego-graphs.
For each instance and each $q$ we generated $30$ realizations $G'$ and ran
\textsc{Pivot} on each. All of these instances are small enough to solve the
min-disagreement LP with Gurobi~13.0.2, so we compare against the LP optimum
of $G'$ throughout.
\Cref{fig:pivot:results,fig:pivot:results} shows the results. The performance of \textsc{Pivot} on all such $G'$ is good enough to indicate that it might behave similarly to the in expectation results of~\Cref{lem:sparsification_G}. 

There are some notable differences between the behavior of \textsc{Pivot} on the planted-clique instances versus ego-graphs instances. We see that for the planted-clique instances
(left side of \Cref{fig:pivot:results}) the approximation ratio increases steadily with
$q$ for all $n$, consistent with the $O(\nicefrac{1}{(1-q)^2})$ dependence of
\Cref{lem:sparsification_G}, even though \textsc{Pivot} is run on $G'$ rather
than on $G$. On the Facebook ego-graphs instances (right side of \Cref{fig:pivot:results}) no such
trend is visible. We suspect this is due to the difference in degree distribution between the
two families: the planted-clique instances are homogeneous by construction,
whereas the Facebook ego-graphs have heavy-tailed degrees, so deletion
affects their structure less uniformly.

\paragraph{Planted-clique instances.} We considered Planted-clique instances for
two reasons. The first is that solving the LP is prohibitive on the larger
social graphs, whereas these instances are small enough that the LP
optimum can be computed. The second is control: social graphs vary
in size, density and community structure at once, whereas here we vary the
number and sizes of the planted cliques.



\bibliography{cc}


\newpage
\phantom{ dirty workaround}

\appendix

\section{Supplementary Experiment Information}

This section contains additional plots for the experiments discussed in Section Experiments.

\begin{table}[!ht]
\centering
\begin{tabular}{lrrr}
\toprule
Graph & \#vertices & \#edges & max.\ pos.\ degree \\
\midrule
FB 3980 & 52 & 292 & 19 \\
FB 698 & 61 & 540 & 30 \\
FB 414 & 150 & 3{,}386 & 58 \\
FB 686 & 168 & 3{,}312 & 78 \\
\midrule
FB 348 & 224 & 6{,}384 & 100 \\
FB 0 & 333 & 5{,}038 & 78 \\
FB 3437 & 534 & 9{,}626 & 108 \\
FB 1912 & 747 & 60{,}050 & 294 \\
FB 1684 & 786 & 28{,}048 & 137 \\
FB 107 & 1{,}034 & 53{,}498 & 254 \\
\bottomrule
\end{tabular}
\caption{Graph statistics for the Facebook ego-networks. The four graphs above
the rule are those on which the LP relaxation was solved exactly; on
the six below it the LP was not solved.}
\label{tab:stats}
\end{table}

\subsection{Min-Max}

 Recall the protocol: for each ego-graph and each
deletion probability $q$ we form $G'$, run our Min-Max Algorithm (\Cref{alg:min-max}) on it, and
record the maximum disagreement. On the four smallest ego-graphs (see \Cref{tab:stats}) we solve
the min-max LP relaxation (\Cref{eq:minmax-lp}) with Gurobi $13.0.2$, so the reported quantity is a
certified approximation ratio against a lower bound on the optimum. \Cref{tab:ratio-lp} gives means, standard deviations and $95\%$ confidence
intervals\footnote{ All Confidence intervals use Student's
$t$ with $\mathrm{df}=\text{Number of independent runs}-1$.} over the $30$ seeds for the first group. \Cref{tab:dmn-ratio} reports the means, standard deviations, and $95\%$ confidence intervals over 30 seeds for the six larger ego-ids (\Cref{tab:stats}). Because the LP is infeasible in these cases, performance is instead compared against the cost reported by~\citet{DMN23} on the undeleted graph, rescaled by $(1-q)$. \Cref{fig:minmax:results-2} shows this plot. Note that the ratios of the algorithm do \emph{not} increase with the deletion probability as one would expect. We suspect that this is an artifact of comparing against the scaled value of ~\citet{DMN23} instead of on actual lower bound on the optimal solution. In~\Cref{fig:minmax:logn}, where we compare against the LP on smaller instances, the ratios scale as expected.

\Cref{alg:min-max} takes a single parameter $\lambda$. We evaluated
$\lambda \in \{5, 8, 12\}$ and
selected the value minimizing the maximum disagreement of the returned
clustering, averaged over seeds; $\lambda = 5$ was best on
most instances and is used throughout. This is below the range
\Cref{lemma:min_max_correlation} requires: the constants there are worst-case
and unoptimised, and larger values only degraded solution quality in our runs.
The experiments therefore report practical performance, not a certified
guarantee.

\begin{equation}\label{eq:minmax-lp}
\begin{alignedat}{2}
\min \quad & \max_{u \in V} \; y_u & \quad & \\[0.4em]
\text{s.t.} \quad
  & y_u = \sum_{v \in N^+_u} x_{uv} + \sum_{v \in N^-_u} (1 - x_{uv})
  && \forall u \in V \\[0.4em]
  & x_{uv} \le x_{uw} + x_{vw} && \forall u,v,w \in V \\[0.2em]
  & 0 \le x_{uv} \le 1 && \forall u,v \in V
\end{alignedat}
\end{equation}

\begin{figure}[t]
    \centering
    \includegraphics[width=0.5\columnwidth]{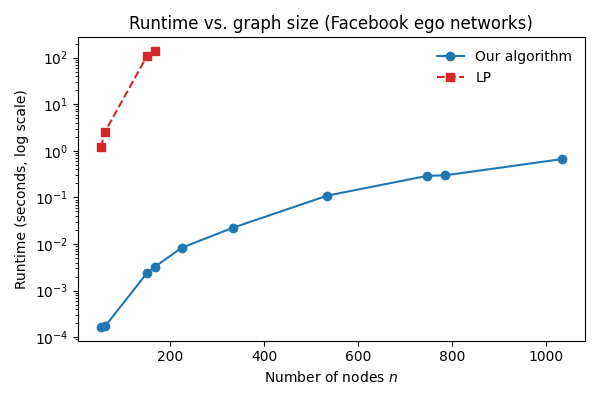}
    \caption{Average running time (log scale) of our combinatorial 
     \Cref{alg:min-max} and of LP on the Facebook ego-networks, as a function of the
    number of vertices $n$.}
    \label{fig:min_max_runtime}
\end{figure}

\begin{figure}[t]
    \centering
    \includegraphics[width=0.5\columnwidth]{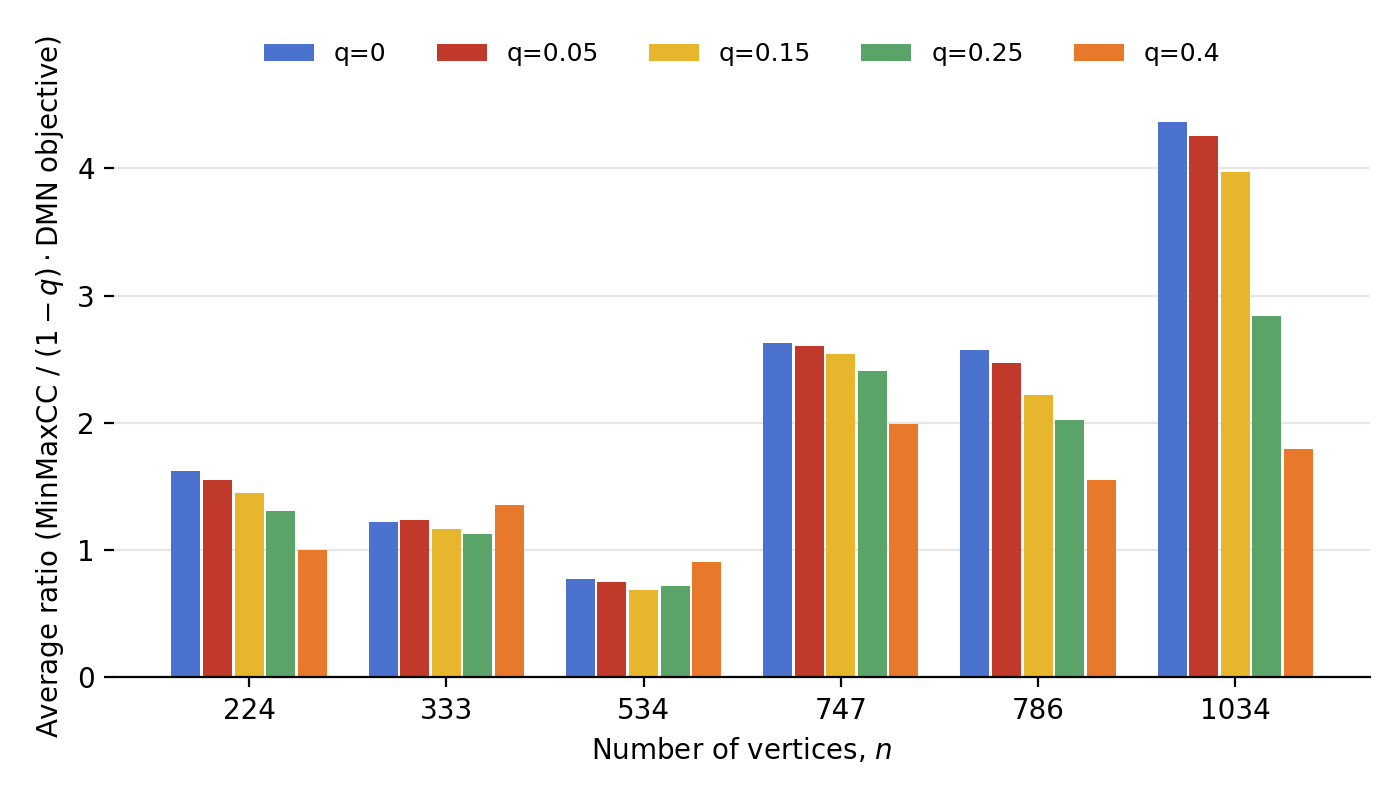}
    \caption{Cost of our Min-Max Algorithm (\Cref{alg:min-max}) on $G'$ relative to $(1-q)$ times the cost
of~\citet{DMN23} on the original graph, over the six biggest ego-graphs where the
min-max LP is infeasible. We ran it with $30$ seeds.}
    \label{fig:minmax:results-2}
\end{figure}

\begin{figure}[t]
    \centering
    \includegraphics[width=0.5\columnwidth]{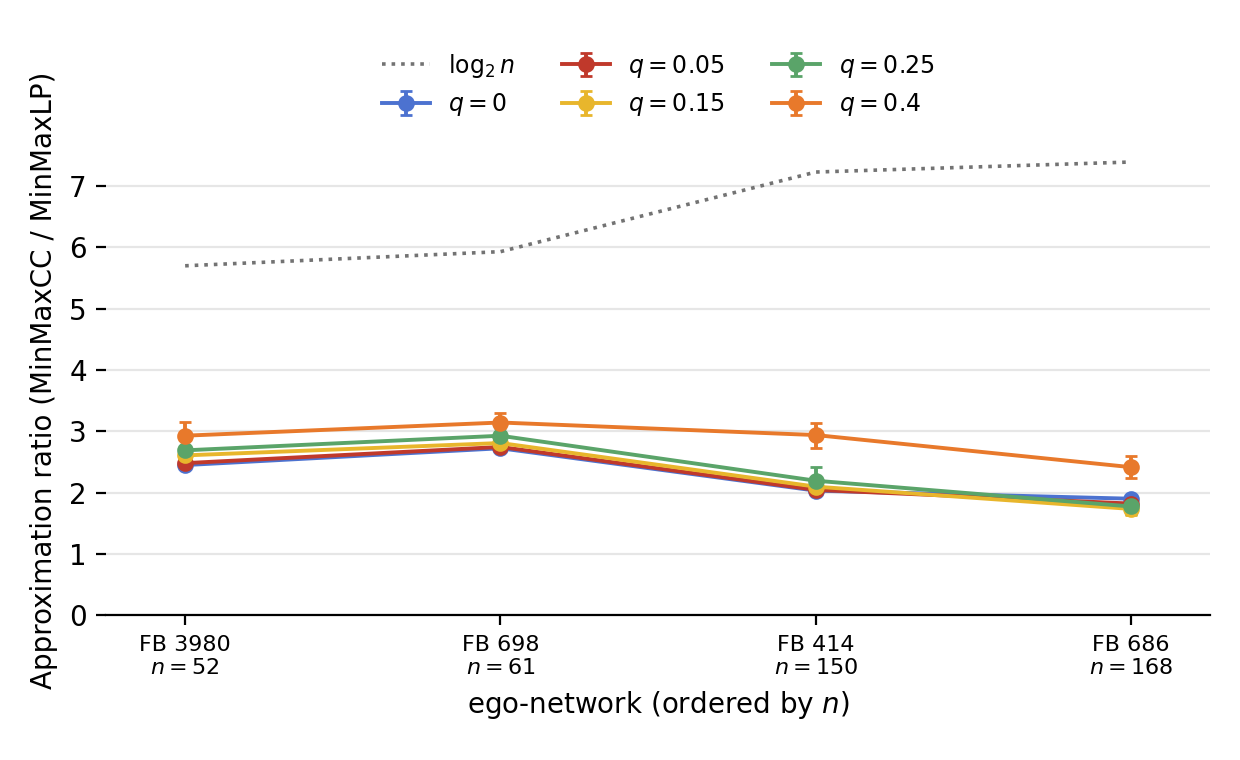}
    \caption{Approximation ratio of our Min-Max Algorithm (\Cref{alg:min-max}) against the min-max LP \Cref{eq:minmax-lp}
    optimum on the four smallest Facebook ego-graphs, plotted against graph
    size; one line per deletion probability $q$}
    \label{fig:minmax:logn}
\end{figure}

\begin{table}[t]
\centering
\small
\begin{tabular}{@{}lcccc@{}}
\toprule
 & \multicolumn{4}{c}{$q$} \\
\cmidrule(l){2-5}
Graph & 0.05 & 0.15 & 0.25 & 0.40 \\
\midrule
FB 348 & 1.550 (0.043) & 1.453 (0.056) & 1.306 (0.088) & 0.998 (0.068) \\
       & {\scriptsize $\pm.016$} & {\scriptsize $\pm.021$} & {\scriptsize $\pm.033$} & {\scriptsize $\pm.025$} \\
\addlinespace
FB 0 & 1.240 (0.038) & 1.170 (0.075) & 1.124 (0.073) & 1.358 (0.118) \\
     & {\scriptsize $\pm.014$} & {\scriptsize $\pm.028$} & {\scriptsize $\pm.027$} & {\scriptsize $\pm.044$} \\
\addlinespace
FB 3437 & 0.747 (0.029) & 0.685 (0.031) & 0.717 (0.041) & 0.904 (0.080) \\
        & {\scriptsize $\pm.011$} & {\scriptsize $\pm.012$} & {\scriptsize $\pm.015$} & {\scriptsize $\pm.030$} \\
\addlinespace
FB 1912 & 2.606 (0.028) & 2.543 (0.031) & 2.405 (0.130) & 1.991 (0.094) \\
        & {\scriptsize $\pm.010$} & {\scriptsize $\pm.012$} & {\scriptsize $\pm.049$} & {\scriptsize $\pm.035$} \\
\addlinespace
FB 1684 & 2.469 (0.043) & 2.223 (0.070) & 2.020 (0.068) & 1.550 (0.200) \\
        & {\scriptsize $\pm.016$} & {\scriptsize $\pm.026$} & {\scriptsize $\pm.025$} & {\scriptsize $\pm.075$} \\
\addlinespace
FB 107 & 4.253 (0.034) & 3.976 (0.033) & 2.844 (0.865) & 1.794 (0.101) \\
       & {\scriptsize $\pm.013$} & {\scriptsize $\pm.012$} & {\scriptsize $\pm.323$} & {\scriptsize $\pm.038$} \\
\bottomrule
\end{tabular}
\caption{Min-max cost of \Cref{alg:min-max} relative to $(1-q)$ times the
min-max cost of~\citet{DMN23} on the undeleted graph, for Facebook
ego-networks under deletion probability $q$; graphs are ordered by size. Each
cell aggregates 30 seeds and gives the mean, the standard deviation in
parentheses, and below it the
$95\%$ confidence interval as $\pm$ a margin. Values below $1$ indicate that our algorithm on $G'$ is
cheaper than the scaled baseline.}
\label{tab:dmn-ratio}
\end{table}

\begin{table}[t]
\centering
\small
\begin{tabular}{@{}lcccc@{}}
\toprule
 & \multicolumn{4}{c}{$q$} \\
\cmidrule(l){2-5}
Graph & 0.05 & 0.15 & 0.25 & 0.40 \\
\midrule
FB 3980 & 2.483 (0.061) & 2.609 (0.118) & 2.692 (0.175) & 2.931 (0.225) \\
        & {\scriptsize $\pm.023$} & {\scriptsize $\pm.044$} & {\scriptsize $\pm.066$} & {\scriptsize $\pm.084$} \\
\addlinespace
FB 698 & 2.747 (0.041) & 2.811 (0.083) & 2.930 (0.122) & 3.145 (0.158) \\
       & {\scriptsize $\pm.015$} & {\scriptsize $\pm.031$} & {\scriptsize $\pm.046$} & {\scriptsize $\pm.059$} \\
\addlinespace
FB 414 & 2.047 (0.052) & 2.099 (0.087) & 2.196 (0.228) & 2.941 (0.204) \\
       & {\scriptsize $\pm.020$} & {\scriptsize $\pm.033$} & {\scriptsize $\pm.085$} & {\scriptsize $\pm.076$} \\
\addlinespace
FB 686 & 1.824 (0.062) & 1.736 (0.101) & 1.776 (0.116) & 2.417 (0.183) \\
       & {\scriptsize $\pm.023$} & {\scriptsize $\pm.038$} & {\scriptsize $\pm.043$} & {\scriptsize $\pm.068$} \\
\bottomrule
\end{tabular}
\caption{Min-Max approximation ratio against the min-max LP (\Cref{eq:minmax-lp}) optimum 
(\Cref{alg:min-max} / MinMaxLP) for Facebook ego-networks under deletion probability
$q$. Each cell aggregates 30 seeds and gives the mean, the standard deviation
in parentheses, and below it the
$95\%$ confidence interval as $\pm$ a margin.}
\label{tab:ratio-lp}
\end{table}
\subsection{Min-disagreement }
We evaluate \textsc{Pivot}~\citep{ACN-pivot} on $G'$ across two instance
families. The first is a planted-clique family: for a given size $n$ we
partition the vertices according to each decomposition in
\Cref{tab:clique_decompositions} and draw a positive label with probability
$0.9$ within a part and $0.1$ across parts. This is a stochastic block
model with $p_{\mathrm{in}} = 0.9$ and $p_{\mathrm{out}} = 0.1$, restricted to
the block structures of \Cref{tab:clique_decompositions}. The second family is
the four smallest Facebook ego-graphs of \Cref{tab:stats}, on which the
min-disagreement LP (\Cref{eq:cc-lp}) is small enough to solve exactly, so ratios can be
measured against a true lower bound. \Cref{tab:fb-pivot-lp} and \Cref{tab:clique-pivot-lp} report the mean and
standard deviation of the observed ratio, together with $95\%$ confidence
intervals, for the Facebook ego-networks and the planted-clique instances
respectively. Figures of the results can be found in the main part of the paper.

\begin{equation}\label{eq:cc-lp}
\begin{alignedat}{2}
\min \quad & \sum_{(i,j) \in E_H^+} x_{ij}
             + \sum_{(i,j) \in E_H^-} (1 - x_{ij}) & \quad & \\[0.3em]
\text{s.t.} \quad
  & x_{ij} \le x_{ik} + x_{jk} && \forall i,j,k \in V \\[0.2em]
  & 0 \le x_{ij} \le 1 && \forall i,j \in V,\; i < j
\end{alignedat}
\end{equation}

\begin{table}[t]
\centering
{\small
\begin{tabular}{cl}
\toprule
\(n\) & \textbf{Clique decompositions} \\
\midrule
10  & \(2\times5,\ (4,3,3),\ (5,3,2)\)                 \\
15  & \(3\times5,\ (5,5,3,2),\ (8,7)\)                 \\
20  & \(2\times10,\ 4\times5,\ (7,7,6)\)               \\
25  & \((10,10,5),\ (12,7,6),\ (13,12),\ (9,8,8)\)     \\
30  & \((15,10,5),\ (20,5,5),\ 2\times15,\ 3\times10\) \\
100 & \(10\times10,\ 4\times25,\ (60,25,10,5)\)        \\
\bottomrule
\end{tabular}
}
\caption{Clique decompositions used to generate the planted-clique instances.
Here, \(k\times s\) denotes \(k\) planted cliques of size \(s\).}
\label{tab:clique_decompositions}
\end{table}

\begin{table}[t]
\centering
\small
\begin{tabular}{@{}rrccccc@{}}
\toprule
 & & \multicolumn{5}{c}{$q$} \\
\cmidrule(l){3-7}
$n$ & Runs & 0.00 & 0.05 & 0.15 & 0.25 & 0.40 \\
\midrule
10 & 150$^{\dagger}$ & 1.731 (0.65) & 1.825 (0.74) & 2.005 (0.87) & 2.189 (1.06) & 2.255 (1.03) \\
   &     & {\scriptsize $\pm.105$} & {\scriptsize $\pm.120$} & {\scriptsize $\pm.141$} & {\scriptsize $\pm.173$} & {\scriptsize $\pm.176$} \\
\addlinespace
15 & 150 & 1.871 (0.22) & 1.960 (0.29) & 2.100 (0.43) & 2.258 (0.64) & 2.575 (1.36) \\
   &     & {\scriptsize $\pm.036$} & {\scriptsize $\pm.048$} & {\scriptsize $\pm.069$} & {\scriptsize $\pm.104$} & {\scriptsize $\pm.219$} \\
\addlinespace
20 & 150 & 1.972 (0.18) & 2.064 (0.26) & 2.253 (0.44) & 2.388 (0.65) & 2.527 (0.83) \\
   &     & {\scriptsize $\pm.029$} & {\scriptsize $\pm.042$} & {\scriptsize $\pm.071$} & {\scriptsize $\pm.105$} & {\scriptsize $\pm.135$} \\
\addlinespace
25 & 200 & 2.068 (0.14) & 2.177 (0.21) & 2.390 (0.35) & 2.524 (0.45) & 2.659 (0.56) \\
   &     & {\scriptsize $\pm.020$} & {\scriptsize $\pm.030$} & {\scriptsize $\pm.049$} & {\scriptsize $\pm.062$} & {\scriptsize $\pm.078$} \\
\addlinespace
30 & 200 & 2.160 (0.12) & 2.304 (0.18) & 2.579 (0.31) & 2.780 (0.42) & 2.996 (0.57) \\
   &     & {\scriptsize $\pm.016$} & {\scriptsize $\pm.025$} & {\scriptsize $\pm.043$} & {\scriptsize $\pm.059$} & {\scriptsize $\pm.080$} \\
\addlinespace
100 & 60 & 2.224 (0.08) & 2.290 (0.16) & 2.394 (0.31) & 2.470 (0.43) & 2.548 (0.59) \\
    &    & {\scriptsize $\pm.021$} & {\scriptsize $\pm.042$} & {\scriptsize $\pm.079$} & {\scriptsize $\pm.110$} & {\scriptsize $\pm.152$} \\
\bottomrule
\end{tabular}
\caption{Ratio of average \textsc{Pivot} cost on $G'$ to the min-disagreement
LP (\Cref{eq:copt-upper}) on the planted-clique instances. Each instance of Table~4 is run
with 50 independent seeds (20 for $n=100$). Cells give the mean, the standard
deviation in parentheses, and below it the
$95\%$ confidence interval as $\pm$ a margin.
$^{\dagger}$At $n=10$, realisations with LP optimum $0$ are omitted, leaving
149, 148 and 133 runs at $q=0.15$, $0.25$ and $0.40$.}
\label{tab:clique-pivot-lp}
\end{table}

\begin{table}[t]
\centering
\small
\begin{tabular}{@{}lcccc@{}}
\toprule
 & \multicolumn{4}{c}{$q$} \\
\cmidrule(l){2-5}
Graph & 0.05 & 0.15 & 0.25 & 0.40 \\
\midrule
FB 3980 & 1.709 (0.030) & 1.692 (0.040) & 1.693 (0.038) & 1.709 (0.048) \\
        & {\scriptsize $\pm.011$} & {\scriptsize $\pm.015$} & {\scriptsize $\pm.014$} & {\scriptsize $\pm.018$} \\
\addlinespace
FB 698 & 1.643 (0.025) & 1.750 (0.042) & 1.839 (0.077) & 1.961 (0.128) \\
       & {\scriptsize $\pm.009$} & {\scriptsize $\pm.016$} & {\scriptsize $\pm.029$} & {\scriptsize $\pm.048$} \\
\addlinespace
FB 414 & 1.670 (0.013) & 1.739 (0.016) & 1.801 (0.023) & 1.880 (0.030) \\
       & {\scriptsize $\pm.005$} & {\scriptsize $\pm.006$} & {\scriptsize $\pm.009$} & {\scriptsize $\pm.011$} \\
\addlinespace
FB 686 & 1.978 (0.011) & 1.955 (0.015) & 1.925 (0.015) & 1.871 (0.014) \\
       & {\scriptsize $\pm.004$} & {\scriptsize $\pm.006$} & {\scriptsize $\pm.006$} & {\scriptsize $\pm.005$} \\
\bottomrule
\end{tabular}
\caption{Ratio of average \textsc{Pivot} cost on the subsampled graph to the
min-disagreement LP (\Cref{eq:cc-lp}), for Facebook ego-networks under deletion
probability $q$. Each cell aggregates 30 independent runs and gives the mean,
the standard deviation in parentheses, and below it below it the
$95\%$ confidence interval as $\pm$ a margin}
\label{tab:fb-pivot-lp}
\end{table}

\end{document}